\documentclass[a4paper,12pt]{article}
\usepackage{amsfonts,amssymb}
\usepackage{latexsym}
\usepackage{mathrsfs}
\usepackage[colorlinks=true,linkcolor=red]{hyperref}

\oddsidemargin
\topmargin
  \usepackage{amsthm}
  \usepackage{amsmath}

  \newtheorem{theorem}{Theorem}[section]
  
  \newtheorem{proposition}[theorem]{Proposition}
  \newtheorem{lemma}[theorem]{Lemma}

  \theoremstyle{definition}
  \newtheorem{definition}[theorem]{Definition}
  \newtheorem{remark}[theorem]{Remark}

  \newcommand{\C}{\mathscr{C}}

  \numberwithin{equation}{section}

  \title{\sc Characterizations of equilibrium allocations in an economy with public goods and infinitely many commodities}
  \author{
{\bf Anuj Bhowmik}\thanks{Indian Statistical Institute, 203 B.T. Road, Kolkata 711108, India. 
 }\\
Indian Statistical Institute\\
\vspace{0.4cm}
e-mail: anujbhowmik09@gmail.com 
}

  \date{\today}

\begin{document}

  \maketitle

\begin{abstract}
This paper examines the characterizations of equilibrium in economies with public projects. Public goods, as discussed by Mas-Colell (1980), are modeled as elements of an abstract set lacking a unified ordering structure. We introduce the concepts of cost share equilibrium for such economies, where the private commodity space is a (possibly nonseparable) Banach lattice. Within this framework, we present two distinct characterizations of cost share equilibria via the veto power of the grand coalition in economies featuring finitely many agents. The first characterization involves allocations that are Aubin non-dominated, while the second establishes that an allocation is a cost share equilibrium if and only if it cannot be dominated by the grand coalition, where domination is considered under specific perturbations of initial endowments.
\end{abstract}



\section{Introduction}

This study explores how cost share equilibrium allocations can be decentralized and examines the role of veto power held by the grand coalition in an economic setting consisting of a finite number of agents, an infinite variety of private commodities, and public projects represented abstractly. In this framework, agent's optimal decisions are influenced by the cost they pay for the public project provision which is not linear in the provision levels.

Public goods are modeled as elements of an abstract set without a universally accepted ordering, a feature stemming from the diverse and often conflicting preferences among individuals. Such circumstances are typical for public goods like national defense, education, and public parks, where people have different perspectives, usage patterns, and subjective rankings. The analysis employs the mathematical framework introduced by Mas-Colell (1980), which allows us to capture this complexity.

By interpreting public projects as “environments”—shared variables externally determined—the framework extends naturally to contexts like legal systems, taxation schemes, and publicly provided private goods. Prior contributions, including Diamantaras and Gilles (1996), Gilles and Diamantaras (1997), Giles and Scotchmer (1997), and Hammond and Villar (1998), illustrate the breadth of this approach. Importantly, the lack of a predefined order on the set of projects rules out the use of conventional monotonicity assumptions, such as those underlying Samuelson’s treatment of public goods preferences.

Mas-Colell (1980) made two fundamental contributions. First, he introduced the concept of a valuation function, which assigns non-linear, individualized prices for access to public goods. Standard taxes or subsidies that agents pay can be interpreted as particular forms of such valuation functions. His first key result established that Pareto optimal allocations can be decentralized through valuation equilibria. His second result demonstrated that, in the case where there is only one private good in the economy, valuation equilibria with nonnegative valuation functions are equivalent to the standard Foley core.

Building on this framework, subsequent works—such as Diamantaras and Gilles (1996), Gilles and Diamantaras (1997), and Giles and Scotchmer (1997)—extended Mas-Colell’s model to economies with multiple private goods. In these extensions, the prices of private commodities are permitted to vary with the level of public project provision. This feature distinguishes their approach from the classical Lindahl equilibrium framework, as it allows for interactions between private prices and public project levels. The motivation stems from the observation that the set of public projects may exhibit non-convexities, and adjustments in provision levels can generate significant changes in prices.

A significant portion of the literature has focused on economies with public projects and infinitely many private commodities. The assumption of infinitely many commodities arises because restricting the model to a finite commodity set implies both a fixed termination date for economic activity and a finite set of uncertain states of the world—assumptions that do not hold in many economic contexts. An infinite-dimensional commodity space is therefore more suitable in models involving an infinite time horizon, an unbounded set of possible states of nature, or infinitely many variations in commodity characteristics. De Simone and Graziano (2004), Basile et al. (2005), and Graziano (2007) studied economies with public projects where the commodity space is an infinite-dimensional commodity space.  All these papers establish the two fundamental theorems of welfare economics under various infinite dimensional commodity space. In this paper, we also posit an economy with public projects and an infinitely dimensional private commodity space. In particular, we consider commodity spaces that are Banach lattices whch may not be necessarily seperable. Our commodity space encompasses a wide range of functional spaces such as the sequence spaces $l_p ( 1\le p \le \infty)$, the Lebesgue spaces $L_p ( 1\le p \le \infty)$ , and the measure space
$M(\Omega)$. This approach allows us to accommodate commodity spaces general enough to include virtually all of the function spaces typically employed in economic analysis (see Mas-Colell and Zame () for a detailed study).

In our economy, the blocking mechanism is modeled through a contribution measure, which assigns to each coalition of agents a cost share corresponding to every possible realization of a public project. More precisely, a contribution measure specifies the portion of the cost that a blocking coalition must bear for a given level of public good provision when attempting to block an allocation.\footnote{For comparison, under the mechanism of Foley (1967), this share is uniformly equal to one for all blocking coalitions.} Complementing this, the cost share function determines the fraction of the public project cost borne by each individual agent, thereby shaping the agent’s budget constraint. Diamantaras and Giles (1996) demonstrate that when cost shares are linear in public projects and represent the Radon–Nikodým derivatives of contribution measures, the core coincides with the set of linear cost share equilibria.

In our framework, we adopt the equilibrium concept of Diamantaras and Giles (1996) but relax the assumption that agents face identical cost shares across all public projects. Instead, aligning with the perspective of Basile et al. (2016), our notion of cost share equilibrium varies according to the individual benefits agents derive from the set of public projects. This formulation captures environments where cost-share functions are better understood as voluntary contributions rather than as fixed tax–subsidy schemes.

This paper considers an economy with finitely many agents, infinitely many commodity space, public projects, and makes two main contributions. As a first main contribution, we establish a core equivalence result for finite economies with public projects and infinitely many private commodities. Addressing a finite set of agents and exchange economies devoid of any public goods, Debreu and Scarf (1963) achieve core equivalence by enlarging the set of coalitions by replicating the original economy, thereby increasing the possibilities of blocking an allocation and hence reducing the set of allocations that are not blocked. A second approach in this regard was suggested by Aubin (1979), where agents are allowed to participate in a coalition with a fraction of their endowments essentially leading to the number of coalitions blocking an allocation to be infinitely large. The resulting veto mechanism is commonly referred to in the literature as the fuzzy, or Aubin, veto mechanism. Aubin (1979) showed that the veto mechanism proposed by him is equivalent to the approach of Debreu and Scarf (1963) and eventually leads to a core that is equivalent to the set of competitive equilibria. In this paper, we adopt the notion of blocking introduced by Aubin (1979) and hence extend the standard idea of a contribution scheme from ordinary coalitions to the broader class of Aubin coalitions (see also Graziano and Romaniello (2011), and 
Bhowmik and Graziano (2015)). Focusing on such an approach, we establish our core equivalence by means of an expanded set of coalitions and a modified veto mechanism thereby genralizing the work of Graziano and Romaniello (2011). Under this modified veto mechanism, the contribution assigned to a blocking coalition toward the realization of a public project accounts explicitly for the proportional participation of each member agent. Further, as a second companion result we explore the veto power of the grand coalition by establishing that for a given cost share function and the corresponding contribution measure, cost share equilibria are exactly those allocations that cannot be blocked by a coalition in which each agent participates with a non-zero fraction of his initial endowment. We also provide a third companion result by drawing a connection to Debreu-Scarf wherein we show that given a any cost share function and the corresponding contribution measure,  an allocation is in the Aubin core if and only if it is a $\sigma$-Edgeworth equilibria. Both our companion results generalizes the companion results of Graziano and Romaniello (2011).

As a second main contribution, we establish a charcterization of the cost share equilibrium allocations in terms of veto power of the grand coalition. Herv\'{e}s-Beloso et al. (2005a) provided a first-ever characterization of equilibrium allocations in terms of veto power of the grand coalition in an economy with asymmetric information and a finite-dimensional commodity space. Such a characterization corrsponds to an extension of the Debreu-Scarf (1963) result but also bears a different flavour. In contrast to Debreu-Scarf who enalrge the number of coalitions that agents can form in order to characterize the equilibrium allocations, Herv\'{e}s-Beloso et al. (2005a) provide a characterize the equilibrium allocations by exercising the veto power of the grand coalition, by enlarging the possible redistribution of the initial endowments.\footnote{In a companion paper,  Herv\'{e}s-Beloso et al. (2005b) generalize the findings of  Herv\'{e}s-Beloso et al. (2005a) to the case of an asymmetric economy where the commodity space is chosen to be the bounded sequence space $\ell^{\infty}$ endowed with Mackey topology. Bhowmik and Cao (2013) further generalize the result of  Herv\'{e}s-Beloso et al. (2003) to an asymmetric information economy where the commodity space is a Banach lattice.} In the presence of public projects, Graziano and Romaniello (2011) showed that an allocation in our finite economy is a cost share equilibrium allocation if and only if it is not blocked by the grand coalition in any of the economies obtained by perturbing the original endowment in the direction of the allocation. 
In such economies, the perturbed endowments vary with each alternative provision of the public good and they achieve their result by suitable modifying the cost function of public projects. In this paper, we extend the findings of  Graziano and Romaniello (2011) to an economy with public projects where the commodity space is a Banach lattice which need not be separable. However, we refrain from modifying the cost functions of public projects and suitably modify the perturbed endowments to warantee the result in our case (see Remark \ref{rem:GR} for a 
comparison with  Graziano and Romaniello (2011)). Note that such a characterization also helps us establish that any non cost share equilibrium allocation is blocked by the grand coalition in a perturbed economy which is very close to the original one.  

However, in order to establish our two main contributions we need to establish core equivalence and Vind's theorem for an associated continuum economy to our finite economy which has an equal treatment nature as already noted in the literature. We construct the associated continuum economy by associating with every agent in the finite economy a continuum of individuals who bears the same consumption set, same utility function and the same cost share function as that of the agent in the finite economy. More formally, given an interval $T =[0,1]$ of agents, we denote by $T = \bigcup\{T_i:i\ge 1\}$, where $T_i = \left[\frac{i-1}{n}, \frac{i}{n}\right)$, if $i \neq n$ and $T_n = \left[\frac{n-1}{n}, 1\right]$ where the set of agents in the finite economy is denoted by $N =\{1,2,\cdots,n\}$, and each agent $t \in T_i$ has the consumption set, utility function, and initial endowment as that of agent $i$ in the finite economy. Aumann (1964) provided a first-ever formalization of a continuum economy by building on the idea of Debreu and Scarf (1963) and established the core equialence theorem for such an economy. In the case of economies with public projects, Diamantaras and Giles (1996) established an equivalence between the set of linear cost share equilibrium allocations and the $\sigma$-core of the economy. We extend the work of Diamantaras and Giles (1996) by establishing an equivalence between the cost share equilibria and the $\sigma$-core to the case of an equal treatment economy with public projects where the commodity space is a Banach lattice. This result plays a central role in establishing the equivalence between the Aubin core and cost share equilibrium allocations in our finite economy. It is worth to mention that thsi result has a significance on its own. Tourky and Yannelis (2001) and Podczeck
(2003) constructed counterexamples of economies to show that the classical core-
Walras equivalence theorem in Aumann (1964) fails whenever the commodity space
is a non-separable ordered Banach space. In their construction, they constructed countably infinitely many utility functions to establish their counterexample. In contrast, we are dealing with finitely many different utility functiions following from equal treatment nature of the economy and thus, we obatin a positive result in our case inspite that the commodity space may not be separable.\footnote{See Bhowmik and Cao (2013) for a simialr result in asymmetric information economies.} A further characterization of the core was established by Schmeidler (1972) and Vind (1972) in an exchange economy with only private commodities. In reality, forming coalitions to block non-core allocation requires communication between individuals, which may be quite costly at times. So characterizing core allocations concerning small coalitions has been studied extensively in the literature. Schmeidler (1972) remarked that blocking coalitions outside the equilibrium allocations is sufficient to consider blocking by arbitrary ``small” coalitions only. Vind (1972) further
generalized Schmeidler’s remark to Aumann’s core-equivalence theorem by considering blocking by coalitions of any given arbitrary size lying between zero and that of the grand coalition. Recognize that this further invokes a normative of the blocking mechanism i.e. the core can be characterized as allocations that survive blocking by a majority of agents. In the context of economies  with public projects, Graziano and Romaniello (2011) provided a first-ever characterization of the $\sigma$-core in line with Vind (1972). In this paper, we extend the work of  
Graziano and Romaniello (2011) to an equal treatment continuum economy with public projects and commodity space with Banach lattice. 
This results plays a central role in establishing characterization of cost share equilibrium allocations in terms of blocing power of the grand coalition.


The paper is organized as follows. Section \ref{sec2} presents the economic model with public projects, introducing the notions of cost-share functions and $\sigma$-core. Section \ref{sec:atomless} focuses on equal treatment economies, establishing the core equivalence theorem and characterizing the $\sigma$-core in terms of the size of blocking coalitions, thus generalizing the results of Schmeidler (1972) and Vind (1972) to this framework. Section \ref{sec5} analyzes cost-share equilibria and proves the two central decentralization results. Finally, Section \ref{appendix} contains the proofs of selected results.

\section{Description of Model}\label{sec2}
In this section, we introduce an exchange economy with (non-Samuelsonian) 
public goods. The set of agents is denoted by $N = \{1,2,\cdots,n\}$. Let $\Sigma = \mathcal P(N) \setminus \emptyset$ denote the set of coalitions in the economy.  The private commodity space in our framework is a Banach lattice $\mathbb B$. Let $\mathbb B_+$ denote the positive cone of $\mathbb B$ whereas $\mathbb 
B_{++}$ denote the interior of  $\mathbb B_+$. It is assumed that each economic agent $i \in N$ is endowed with 
quantities of private goods, denoted by $e_i$, such that 
$\sum_{i \in N} e_i \in \mathbb B_{++}$.

\medskip
There is a set $\mathscr Y$ of public projects, on which we do not impose any order or mathematical structure. 
An element $y\in \mathscr Y$, referred to as the {\bf provision level} of the public goods, can be interpreted as the representation of some uniform quantity level of provision of the public goods. The costs for the realization of  
public projects in terms of private goods are expressed by a {\bf cost function} $c:\mathscr Y\to \mathbb B_+$, 
where $c(y)$ denotes the required quantities of private goods to provide the economy with the level $y$  of 
public goods. 

\medskip
To complete the description of an economy, we assume that each agent $i\in N$ has a rational preference 
relation defined on $\mathbb B_+\times \mathscr Y$, which is represented by a utility function $u_i:\mathbb 
B_+\times \mathscr Y\to \mathbb R$. We call the collection $\{u_i:i\in N\}$ of utility functions {\bf desirable} for the economy $\mathscr E$
if it satisfies the following properties: 

\begin{enumerate}
\item[(i)] {\bf Continuity}: The utility function $u_i:\mathbb B_+\times \mathscr Y\to \mathbb R$ is {\bf continuous} 
if $u_i(\cdot, z)$ is continuous for all $z\in \mathscr Y$. 

\item[(ii)] {\bf Monotonicity}:  The utility function $u_i:\mathbb B_+\times \mathscr Y\to \mathbb R$ is {\bf strongly monotone} 
 in the sense that for all $\xi, \zeta\in \mathbb B_{+}$ and all $z\in \mathscr Y$ with $\xi\ge \zeta$ with $\xi\neq \zeta$, we have $u_i(\xi, z)> u_i(\zeta, z)$. 

\item[(iii)] {\bf Quasi-concavity}: For each $i \in N$, the utility function $u_i:\mathbb B_+\times \mathscr Y\to \mathbb R$ is 
{\bf quasi-concave} in the sense that $u_i(\cdot, z)$ is quasi-concave for all $z\in \mathscr Y$.


\end{enumerate}

\begin{definition}\label{defn:economy}
An {\bf economy with (non-Samuelsonian) public goods} is a collection $\mathscr E:=\{N, 
\mathbb B_+, (\mathscr Y, c), (u_i, e_i)_{i\in N}\}$ satisfying the following properties:
\begin{enumerate}
\item[(a)] The collection $\{u_i: i \in N\}$ is desirable; and 

\item[(b)] The total initial endowment ${\bf e}:=\sum_{i \in N} e_i$ satsifies the inequality ${\bf e}-c(y)\in 
\mathbb B_{++}$, for all $y\in \mathscr Y$. This condition ensures that each private commodity is 
present on the market regardless of the cost of the project that is going to be realized. 
\end{enumerate}
\end{definition}

\medskip
\noindent
An {\bf allocation} for the economy $\mathscr E$ is a pair $(x_1,\cdots,x_n, y)$, where $x_i$ specifies the amount of private goods assigned to the agent $i$, and $y\in 
\mathscr Y$ is a public project. Further, an allocation $(x_1, \cdots, x_n, y)$ is said to be {\bf feasible} if 
\[
\sum\limits_{i \in N} x_i + c(y)= \sum\limits_{i \in N} e_i.
\]
This means that the initial endowment is not used to cover the cost of the realized project is 
redistributed among the agents. We denote by $\mathcal A$ the set of all feasible allocations in the economy. Define 
\[
\mathcal A^p:=\{(x_1,\cdots,x_n): (x_1, x_2,\cdots, x_n, y) \in \mathcal A \mbox{ for some } y\in \mathscr Y\}.
\]
Thus, an element of $\mathcal A^p$ is just an allocation of private goods among the agents. We now introduce the notion of a cost distribution function, which describes how much each economic agent must contribute to the establishment of any 
provision level of the public projects. This notion shall come in handy when we define the concept of a cost share equilibrium as introduced in Diamantaras and Giles (1997). We say that a function $\rho:N \times \mathscr Y\to [0, 1]$ is a {\bf cost distribution
function} if the partial function $\rho(\cdot, y)$ is such that $\sum\limits_{i \in N} \rho(i,y) =1$ 
for all $y\in \mathscr Y$. We denote by $\mathscr D$ the set of all 
cost distribution functions for the economy $\mathscr E$. Notice that, similar to Basile et al. (2016), 
we assume that the individual contribution may vary across public projects.


\medskip
We now introduce the notion of a cost share equilibrium. To do this, let $\mathbb 
B^*$ denote the norm-dual of $\mathbb B$, whereas $\mathbb B^*_+$ is the positive cone 
of $\mathbb B^*$. A {\bf price system} is a function $\pi:\mathscr Y\to \mathbb B_+^*$.  Then given a system of prices for private commodities, the notion of cost distribution 
function, describes how much each economic agent must contribute to the establishment of any 
provision level of the public projects. Therefore, if $\rho\in \mathscr D$ is a cost distribution function, 
then agent $i$ is expected to contribute the amount $\rho(i,y)\pi(y)\cdot c(y)$, 
where $\pi(y)$ is the market price vector for private commodities.

\begin{definition}
 A feasible allocation $(x_1,\cdots,x_n, y)$ is a {\bf cost share equilibrium} if there exists a price system 
$\pi:\mathscr Y\to \mathbb B_+^*$ and a cost distribution function $\rho$ such that, for 
all $i \in N$, $(x_i, y)$ maximizes the utility $u_i$ on the budget set 
\[
\mathbb B_i(\pi):=\left\{(\zeta, z)\in \mathbb B_+\times \mathscr Y: \pi(z)\cdot \zeta + \rho(i,z)\pi(z)\cdot c(z)\le \pi(z)\cdot e_i\right\}.
\]
The set of cost share equilibrium of the economy $\mathscr E$ is denoted by ${\rm CE}(\mathscr E)$. 
\end{definition}

Let $\rho\in \mathscr D$ be a cost distribution function. We denote by ${\rm CE}_\rho(\mathscr E)$ the 
collection of all cost share equilibria whose corresponding cost distribution function is equal to $\rho$. 
Therefore, we have the following identity:
\[
{\rm CE}(\mathscr E)=\bigcup\{{\rm CE}_\rho(\mathscr E):\rho\in \mathscr D\}.
\]
It is important to acknowledge that, for any given level of public project provision, the cost distribution among agents in a 
cost share equilibrium may vary; it is not necessarily constant. This gives rise to the concept of an {\bf equal cost share 
equilibrium}: a cost share equilibrium in which the cost distribution function $\rho:N\times \mathscr Y
\to [0, 1]$ is such that 
\[
\rho(i, z)=\frac{c(z)}{n},
\] 
for every agent $i \in N$ and every provision level $z\in \mathscr Y$. A different specification is the {\bf linear cost share equilibrium}, 
introduced by Diamantaras and Gilles (1997), which requires that $\rho(i,z)$ be independent of the provision level
$z$ for all agent $i \in N$.  

\medskip
The utility maximization behavior exhibited by agents in a cost share equilibrium closely mirrors that of standard competitive 
equilibrium models involving only (infinitely many) private goods. However, it diverges significantly from the framework of 
Lindahl equilibrium typically applied in public goods contexts. In a cost share equilibrium, the price system that defines each 
agent’s budget set is contingent upon the entire set of public projects. Although only one public project is ultimately implemented, the price system 
$\pi: \mathscr Y \to \mathbb B^*_+$ can be interpreted as capturing the full spectrum of possible shifts in the private goods 
sector induced by varying public good provisions.

\medskip
This framework notably departs from the classical Samuelsonian model of public goods, primarily due to the absence of monotonic preferences over public projects—since the set 
$\mathscr Y$ lacks an inherent order. Even when $\mathscr Y$ is endowed with a linear structure, decentralizing optimal allocations without a price system that is 
contingent on the public projects may be impossible (see Diamantaras et al. (1996)).

\medskip
In contrast to the standard competitive equilibrium, the concept of a cost share equilibrium is characterized by the inclusion of cost distribution functions. Specifically, the term 
$\rho(i,y)\pi(y)\cdot c(y)$ within agent $i$'s budget constraint represents the individual price they must pay to access the consumption of public project 
$y\in \mathscr Y$. This term mirrors the personalized pricing mechanism found in Lindahl equilibrium models, where individuals pay for public goods according to their marginal benefits. However, unlike in Lindahl equilibrium, these personalized prices in a cost share equilibrium are contingent upon the prices of private goods. Notably, when the set of public projects 
$\mathscr Y$ contains only a single project ($|\mathscr Y|=1$), the notions of cost share equilibrium and linear cost share equilibrium coincide, as the cost distribution becomes uniform across all agents.

\medskip
In what follows, to model the veto mechanism underlying the notion of core, which is most compatible with the notion 
of cost share equilibrium allocations, we assume that each potentially blocking coalition bears a fixed share of the total 
cost of the public project. Crucially, this fixed cost share may differ across projects and need not scale with the 
coalition’s size. The concept of a contribution measure introduced here succinctly encapsulates this flexibility.

\begin{definition}
A {\bf contribution measure} is a function $\sigma:\Sigma\times \mathscr Y\to [0, 1]$ such that 
for each $y\in \mathscr Y$, the partial function $\sigma(\cdot, y): \Sigma \to [0, 1]$ is an additive function such that $\sigma(N, y)=1$. We denote by
$\mathscr M$ the collection of all contribution measures for $\mathscr E$. 
\end{definition}


\begin{definition}
Given an $\sigma\in \mathscr M$, an allocation $(x_1,\cdots,x_n, y)$ is {\bf $\sigma$-blocked by a coalition $S$}
if there exists a public project $z\in \mathscr Y$ and an allocation $(\xi_1,\cdots,\xi_n)$ of private commodities such that 
\begin{enumerate}
\item[(i)] $u_i(\xi_i, y)> u_i(x_i,y)$ for all $i\in S$; and
\item[(ii)] $\sum_{i\in S}\xi_i+ \sigma(S, z) c(z)=\sum_{i\in S} e_i$. 
\end{enumerate}
The {\bf $\sigma$-core} of $\mathscr E$, denoted by ${\C}^{\sigma}(\mathscr E)$, comprises all feasible allocations that 
cannot be $\sigma$-blocked by any coalition. This concept, introduced by Basile et el. (2016), generalizes the classical 
notion of the $\sigma$-core--originally developed by Diamantaras and Gilles (1996) --where the contribution measure is
 invariant with respect to the provision of the public project. In the particular case where the contribution measure
$\sigma(\cdot, y)$ coincides with $\mu$ for any $y\in \mathscr Y$, the corresponding core is referred to 
as the {\bf proportional core} of $\mathscr E$ and is denoted by 
${\C}^\mu(\mathscr E)$. 
\end{definition}

\begin{definition}
A feasible allocation $(x_1,\cdots, x_n, y)$ is said to be {\bf dominated} by another feasible allocation $(\xi_1,\cdots,\xi_n, z)$ if 
$u_i(\xi_i, z)> u_i(x_i,y)$ for every $i \in N$. A feasible allocation is {\bf Pareto optimal} whenever 
it is not dominated by any other feasible allocation.   
\end{definition}
Recognized that the notion of Pareto optimality does not depend on any contribution measure, 
and for any $\sigma\in \mathscr M$, any allocation in the corresponding $\sigma$-core is also a
Pareto optimal allocation.

\medskip
In what follows, we introduce two key concepts: the essentiality condition and the integrability property of preferences. The former, introduced by Diamantaras and Gilles (1997), ensures that private goods are valued in a fundamental way by agents. The latter, as developed in Graziano and Romaniello (2012), imposes a stronger requirement—essentially mandating the interchangeability in the provision of public goods. These two conditions will be implicitly assumed throughout our results, even when not stated explicitly.

\medskip
\noindent
{\bf Essentiality condition:} The economy $\mathscr E$ is said to satisfy the {\bf essentiality condition} if the following two conditions 
hold:

\begin{itemize}
\item[(i)] For every allocation $(x_1,\cdots,x_n,y)$, every provision level $z$ and for each agent 
$i \in N$, there is some $\xi\in \mathbb B_+$ such that $u_i(\xi, z)\ge u_i(x_i,y)$. 

\item[(ii)] For each agent $i \in N$, every $\xi\in \mathbb B_+$ and all provision 
levels $y, z\in \mathscr Y$, we have $u_i(\xi,y)\ge u_i(0,z)$. 
\end{itemize}

The first essentiality condition ensures that any variation in the public goods
provision can be compensated by a suitable quantity of private goods. In the second condition, setting
$u_i(0, z) = 0$ is just a normalization; the important part of this latter statement is that
$u_i(0, y) = u_i(0, z)$ for all projects $y, z\in \mathscr Y$. This is similar to the indispensability
condition of Mas-Colell (1980).

\medskip
\noindent
{\bf  Integrable Utilities:} For any feasible allocation $(x_1,\cdots,x_n, y)$ and all $z\in \mathscr Y$ there is 
an element $(\xi_1,\cdots,\xi_n) \in \mathcal A^p$ such that $(\xi_1,\cdots,\xi_n, z)$ is feasible and $u_i(\xi_i, z)\ge u_i(x_i, y)$
for all $i \in N$, whenever $U(x_i,y):=\{\zeta \in \mathbb B_+:u_i(\zeta, z)> u_i(x_i, y)\}\neq \emptyset$ for all $i \in N$.

\begin{remark}\label{rem:gamma}
Let $(x_1,\cdots,x_n,y)$ be a feasible allocation and $z\in \mathscr Y$. By the essentiality and monotonicity conditions, we have 
\[
U(x_i,y):=\{\zeta \in \mathbb B_+:u_i(\zeta, z)> u_i(x_i, y)\}\neq \emptyset,
\] 
for all $i \in N$. Therefore, as the utities are integrable, we can find an element $(\gamma_1^z,\cdots,\gamma_n^z) \in \mathcal A^p$ such that 
$(\gamma_1^z,\cdots,\gamma_n^z, z)$ is feasible and $u_i(\gamma_i^z, z)\ge u_i(x_i, y)$, for all $i \in N$. Futher, if  $(x_1,\cdots,x_n,y)$ 
is a cost share equilibrium under the price system $\pi:\mathscr Y\to \mathbb B^*_+$ then it can be readily verified that 
\[
\pi(z)\cdot \gamma^z_i+ \rho(i,z)\pi(z)\cdot c(z)= \pi(z)\cdot e_i,
\]
for all $i\in N$.

\end{remark}





\section{An Equal Treatment Continuum Economy}\label{sec:atomless}
In this section, we aim to characterize cost share equilibria in terms of $\sigma$-core allocations of an equal treatment continuum economy, which is derived from a finite economy.
We begin by constructing a continuum economy $\mathscr E_c$ with $n$ different types of agents canonically associated with the finite economy $\mathscr E$, associating each agent 
$i\in N$ with a continuum $T_i$ of individual economic agents. Each agent in $T_i$ shares the same type as agent $i$ in $\mathscr E$ and has a negligible influence on the market. This construction follows standard procedures analogous to those used in economies with only private commodities, as detailed in the literature (see for example Garc\'{i}a-Cutr\'{i}n and Herv\'{e}s-Beloso (1993), Herv\'{e}s-Beloso et al. (2005a,b), Graziano and Romaniello (2012), and Bhowmik and Cao (2013) among others for a similar construction). Moreover, it establishes a correspondence between the cost share functions (and also contribution measures) between the finite and continuum economies.


\medskip
The space of agents is represented by the interval $T = [0,1]$ and is endowed with the Lebesgue $\sigma$-algebra $\mathscr T$ and the Lebesgue measure $\mu$. We denote by $T = \bigcup\{T_i:i\in N\}$, where 
\[
T_i = \left[\frac{i-1}{n}, \frac{i}{n}\right), \mbox{ if } i \neq n \mbox{ and } T_n = \left[\frac{n-1}{n}, 1\right].
\] 
Each agent $t \in T_i$ has a consumption set $\mathbb B_+$, a utility function $u_t = u_i$, and an initial endowment $e(t) = e_i$. Thus, $T_i$ 
denotes the set of agents of type $i$  in the economy $\mathscr E_c$. The set of public projects is denoted by $\mathscr Y$ with the function 
$\widehat{c}: \mathscr Y \to \mathbb B_+$ defined as $\widehat{c}(y) = \frac{c(y)}{n}$, representing the cost function. Notice that, for all $y\in \mathscr Y$,  
\[
\int_T e\, d\mu-\widehat c(y)=\frac{1}{n}\left(\sum_{i\in N}e_i - c(y)\right) \in \mathbb B_{++}.
\]  
An {\bf allocation} for the economy $\mathscr E_c$ is a pair $(f, y)$, where $f(t)$ specifies the amount of private goods assigned to the agent $t$, and $y\in 
\mathscr Y$ is a public project.
We call an allocation $(f,y)$ to be {\bf feasible} in the economy $\mathscr E_c$ if \[
\int_T f\, d\mu + c(y) = \int_T e\, d\mu.
\] 
Further, one can observe that an allocation $(x_1,\cdots,x_n,y)$ in $\mathscr E$ can be identified as an allocation $(f,y)$ in $\mathscr E_c$, where $f$ is the function defined as $f(t) = x_i$ if $t \in T_i$. Reciprocally, an allocation $(f,y)$ in $\mathscr E_c$ can be identified as an allocation $(x_1,\cdots, x_n,y)$ in $\mathscr E$, where $x_i = \frac{1}{\mu(T_i)}\int_{T_i}f \, d\mu$.

\medskip
\begin{definition}
We say that a function $\widehat{\rho}:T \times \mathscr Y\to \mathbb R_+$ is a {\bf cost distribution
function} in the economy $\mathscr E_c$ if the partial function $\rho(\cdot, y)$ is such that $\int_{T} \rho(\cdot,y) d\mu =1$ 
for all $y\in \mathscr Y$. We denote by $\mathscr D_c$ the set of all cost distribution functions in $\mathscr E_c$.
 A feasible allocation $(f, y)$ is a {\bf cost share equilibrium} if there exists a price system 
$\pi:\mathscr Y\to \mathbb B_+^*$ and a cost distribution function $\widehat{\rho}\in \mathscr D_c$ such that, for 
almost all $t\in T$, $(f(t), y)$ maximizes the utility $u_t$ on the budget set 
\[
\mathbb B_t(\pi):=\left\{(\zeta, z)\in \mathbb B_+\times \mathscr Y: \pi(z)\cdot \zeta + \widehat{\rho}(t,z)\pi(z)\cdot c(z)\le \pi(z)\cdot e(t)\right\}.
\]
The set of cost share equilibrium for the economy $\mathscr E$ is denoted by ${\rm CE}(\mathscr E_c)$. 
\end{definition}

\medskip
We now introduce he concept of the $\widehat{\sigma}$-core for the economy $\mathscr E_c$. Before doing so, we begin with the definition of a coalition in 
$\mathscr E_c$. A \textbf{coalition} is defined as a measurable subset $B\subseteq T$ with strictly positive measure.

\begin{definition}
A {\bf contribution measure} in $\mathscr E_c$ is a function $\widehat{\sigma}:\mathscr T\times \mathscr Y\to [0, 1]$ such that 
for each $y\in \mathscr Y$, the partial function $\widehat{\sigma}(\cdot, y):\mathscr T\to [0, 1]$ is a measure 
that is absolutely continuous with respect to $\mu$ such that $\widehat{\sigma}(T, y)=1$. We denote by
$\mathscr M_c$ the collection of all contribution measures in $\mathscr E_c$. 
Given an $\widehat{\sigma}\in \mathscr M_c$, an allocation $(f, y)$ is {\bf $\widehat{\sigma}$-blocked by a coalition $S$}
if there exists a public project $z\in \mathscr Y$ and an allocation $g$ of private commodities such that 
\begin{enumerate}
\item[(i)] $u_t(g(t), z)> u_t(f(t),y)$ $\mu$-a.e. on $S$; and
\item[(ii)] $\int_S g \; d\mu+ \sigma(S, z) c(z)=\int_S e \; d\mu$. 
\end{enumerate}
The {\bf $\widehat{\sigma}$-core} of $\mathscr E_c$, denoted by ${\C}^{\widehat{\sigma}}(\mathscr E_c)$, is the set of all feasible 
allocations that cannot be $\widehat{\sigma}$-blocked by any coalition. 
\end{definition}

\begin{definition}
A feasible allocation $(f, y)$ is said to be {\bf dominated} by another feasible allocation $(g, z)$ if 
$u_t(g(t), z)> u_t(f(t),y)$ $\mu$-a.e. on $T$. A feasible allocation is {\bf Pareto optimal} whenever 
it is not dominated by any other feasible allocation.   
\end{definition}

\begin{remark}

There exists a one-to-one correspondence between \(\mathscr{D}_c\) and \(\mathscr{M}_c\):  

\medskip

(i) For each \(\widehat{\rho} \in \mathscr{D}_c\), there is a unique contribution measure \(\widehat{\sigma}_{\widehat{\rho}} \in \mathscr{M}_c\) associated with \(\widehat{\rho}\), defined by  
\[
\widehat{\sigma}_{\widehat{\rho}}(S, y) := \int_S \widehat{\rho}(\cdot, y)\, d\mu, \quad \text{for every } S \in \mathscr{T}.
\]  

(ii) Conversely, for each \(\widehat{\sigma} \in \mathscr{M}_c\), there exists a Radon–Nikodym derivative \(\widehat{\rho}_{\widehat{\sigma}}(\cdot, y) \in \mathscr{D}_c\) of \(\widehat{\sigma}(\cdot, y)\) with respect to \(\mu\), for all \(y \in \mathscr{Y}\).

\end{remark}

\begin{remark}
If $\rho$ is a cost distribution function in $\mathscr E$, then the function $\widehat{\rho}: T\times \mathscr Y \to \mathbb R_+$, defined as $\widehat{\rho}(t, z) = 
n\rho(i,z)$, represents a cost distribution function for the economy $\mathscr E_c$. Reciprocally, given a cost distribution function $\widehat{\rho}$ in the economy 
$\mathscr E_c$, the function $\rho: N\times \mathscr Y\to [0,1]$ identifies a cost distribution function for the economy $\mathscr E$, where 
$\rho(i,z) = \int_{T_i}\widehat{\rho}(\cdot, z)\, d \mu$ for all $i\in N$ and all $z\in \mathscr Y$.
\end{remark}

\begin{remark}
If $\sigma$ is a contribution measure for $\mathscr E$, then $\widehat{\sigma}: \mathscr T\times \mathscr Y \to [0,1]$, defined as 
\[
\widehat{\sigma}(S,z) = \sum\limits_{i = 1}^{n} \sigma(\left\{i\right\},z)\frac{\mu(S \cap T_i)}{\mu(T_i)},
\]
represents a contribution measure for $\mathscr E_c$. On the other hand, given a contribution measure $\widehat{\sigma}$ for $\mathscr E_c$, the function 
$\sigma: \Sigma \times \mathscr Y\to [0,1]$, defined by $\sigma(S, z)=\sum_{i\in S} \widehat \sigma(T_i, z)$, identifies a contribution measure for $\mathscr E$.
\end{remark}

\subsection{Equivalence Theorem}\label{sec:equi}
In this subsection, we explore the decentralization of the core by characterizing it through the cost share equilibria. It is straightforward 
to verify that, in economies with public projects where the commodity space is a Banach lattice, every cost share equilibrium lies within the 
$\widehat{\sigma}$-core (see Graziano and Romaniello (2012) for a proof in the case of finitely many commodities). Our objective here is 
to establish the converse. Specifically, we show that any allocation in the $\widehat{\sigma}$-core of the economy $\mathscr{E}_c$ that satisfies the 
equal treatment property must be a cost share equilibrium. This result is derived under three alternative assumptions: (i) the positive cone has a nonempty 
interior; (ii) the positive cone lacks an interior but has a quasi-interior point; and (iii) the positive cone has no quasi-interior point.

\begin{theorem}\label{thm:equivalence}
  Let $\mathscr{E}_c$ be an equal treatment economy with public goods, constructed from a finite economy $\mathscr{E}$ such that ${\rm int}\; \mathbb B_+ \neq \emptyset$. Let $\sigma$ be a contribution measure for $\mathscr{E}$, and let $\widehat{\sigma}$ denote the contribution measure induced by $\sigma$ in $\mathscr{E}_c$.  Suppose $(f, y)$ is a feasible allocation in $\mathscr E_c$ such that $f(t) = x_i$ for all $t \in T_i$ and $i \in N$, and $(f, y)$ belongs to the $\widehat{\sigma}$-core of $\mathscr{E}_c$. Then, there exists a non-zero price system $\pi: \mathscr{Y} \to \mathbb{B}^*_+$ such that $(f, y)$ is a cost share equilibrium of $\mathscr{E}_c$.
  \end{theorem}

\begin{proof}
Suppose that $\widehat{\rho}$ (resp. $\rho$) is the cost distribution function in $\mathscr{E}_c$ (resp. $\mathscr E$) associated 
with the contribution measure $\widehat{\sigma}$ (resp. $\sigma$). Therefore, $\widehat{\rho}(t, z)=\rho(i,z)$ for all $t\in T_i$ with $i\in N$ and all 
$z\in \mathscr Y$. For any $z\in \mathscr Y$ and any $t\in T$, we define 
\[
{\bf G} (t, z):= \left\{g\in \mathbb B_{+}: u_t(g, z)> u_t(f(t), y)\right\}
\]
and 
\[
{\bf F} (t, z):= \left({\bf G} (t, z)-e(t)\right)\cup \left\{-\widehat{\rho}(t,z)c(z)\right\}.
\]
Finally, for any $z\in \mathscr Y$, we define
\[
{\bf F}(z):= {\rm cl}\int_T {\bf F}(\cdot, z) d\mu+ c(z).
\]

\noindent
{\bf Claim 1.} \emph{For any $z\in \mathscr Y$, we have ${\bf F}(z)\cap (-\mathbb B_{++})=\emptyset$.} 
Assume that the claim is false. Then, as $\mathbb B_{++}$ is open, there exist an $z\in \mathscr Y$
 and an integrable selection $h$ of ${\bf F}(\cdot, z)$ such that 
\[
\int_T h\, d\mu+c(z)\in -\mathbb B_{++}.
\]
Define 
\[
R:=\left\{t\in T: h(t)\neq -\widehat{\rho}(t,z)c(z)\right\}.
\] 
It follows that $\mu(R)> 0$ and 
\[
\int_R h\, d\mu+\widehat{\sigma}(R, z)c(z)\in -\mathbb B_{++}.
\]
Furthermore, there is an $\varphi(t)\in {\bf G} (t, z)$ such that 
$h(t):=\varphi(t)-e(t)$ for all $t\in R$. Let $\xi\in \mathbb B_{++}$ be such that 
\[
\int_R (\varphi-e)d\mu+\sigma(R, z)c(z)=-\xi.
\] 
Define a function $\psi: T \to \mathbb B_+$ by letting $\psi(t):=\varphi(t)+\frac{\xi}{\mu(R)}$, for all $t\in R$. By 
the strong monotonicity of preferences with respect to private commodities, we conclude that 
$(f,y)$ is $\widehat{\sigma}$-blocked by $R$ via $(\psi, z)$. This is a contradiction.

\medskip
\noindent
{\bf Claim 2:} $\pi(z)  \in \mathbb B^*_{++}$. By the separating hyperplane theorem, there is a price system 
$\pi:\mathscr Y\to \mathbb B^*$ such that $\pi(z)\cdot \psi\ge 0$ for all $\psi\in {\bf F}(z)$ and all $z\in \mathscr Y$. 
It can be readily verified that $\pi(z)\in \mathbb B^*_+$ for all $z\in \mathscr Y$.
We know that  for any $z \in \mathscr Y$, $c(z) \ll \sum_{i \in N}e_i$. Thus, it follows that $\pi(z) \cdot c(z) < \pi(z) \cdot \sum_{i \in N} e_i$ which can be written as 
\[
\sum_{i \in N}\rho(i,z) \pi(z)\cdot c(z) < \pi(z) \cdot \sum_{i \in N} e_i.
\] 
Then there must exist some $i_0 \in N$ such that $\rho(i_0, z) \pi(z)\cdot c(z) < \pi(z) \cdot e_{i_0}$.

\medskip
\noindent
{\bf Step I:}  Choose an element $(g_{i_0},z)$ such that $u_{i_0}(g_{i_0}, z) > u_{i_0}(x_{i_0},y)$. We show that $\pi(z) \cdot g_{i_0} + \rho(i_0, z) \pi(z)\cdot c(z) > \pi(z) \cdot e_{i_0}$. To this end, we assume that $\pi(z) \cdot g_{i_0} + \rho(i_0, z) \pi(z)\cdot c(z) \le \pi(z) \cdot e_{i_0}$ holds. Since 
$\rho(i_0, z) \pi(z)\cdot c(z) < \pi(z)\cdot e_{i_0}$ it follows that $\pi(z) \cdot g_{i_0} > 0$. Consequently, there is some $\alpha \in (0,1)$ such that 
$u_{i_0}(\alpha g_{i_0}, z) > u_{i_0}(x_{i_0}, y)$ and 
\[
\pi(z) \cdot \alpha g_{i_0} + \rho(i_0, z) \pi(z)\cdot c(z) < \pi(z) \cdot e_{i_0}.
\]
 Define an allocation of private goods $\hat h: T \to \mathbb B_+$ as
\[
\hat{h}(t): = \left\{
\begin{array}{ll}
\hspace{3pt} \alpha g_{i_0}-e_{i_0}, & \mbox{ if $t \in T_{i_0}$; }\\
-\rho(i, z)c(z), \hspace{2pt} & \mbox{if $t \in T_i$, $i \neq i_0$.}
\end{array}
\right.
\]
We then obtain that
\[
\int_{T} \hat h\, d\mu + c(z)  = (\alpha g_{i_0} - e_{i_0})\mu(T_{i_0}) + \rho(i_0, z)c(z)\mu(T_{i_0}).
\]
Therefore, 
\[
\pi(z) \cdot \left[\int_{T} \hat h\, d\mu + c(z)\right]  =\pi(z) \cdot (\alpha g_{i_0} - e_{i_0})\mu(T_{i_0}) + \rho(i_0, z)\mu(T_{i_0}) \pi(z) \cdot c(z)< 0.
\]
This is a contradiction. Thus, $\pi(z) \cdot g_{i_0} + \rho(i_0, z) \pi(z) \cdot c(z) > \pi(z) \cdot e_{i_0}$. 

\medskip
\noindent
{\bf Step II:} Fix an element $z\in \mathscr Y$. By Remark \ref{rem:gamma}, there exists a feasible allocation $(\gamma_1^z,\cdots, \gamma_n^z, z)$ 
such that $u_i(\gamma_i^z,z) \ge u_i(x_i,y)$ for all $i\in N$. In what follows, we verify that 
\[
\pi(z) \cdot \gamma_i^z + \rho(i,z) \pi(z) \cdot c(z) = \pi(z) \cdot e_i
\]
for all $i\in N$. To this end, for each $\varepsilon > 0$, define $\delta_i^\varepsilon: = \gamma_i^z + \varepsilon$.  By the monotonicity of preferences, we have 
$u_i(\delta_i^\varepsilon, z)> u_i(\gamma_i^z,z) \ge u_i(x_i,y)$ for all $i\in N$. We show that 
\[
\pi(z) \cdot \delta_i^\varepsilon + \rho(i,z) \pi(z) \cdot c(z) \ge \pi(z) \cdot e_i
\]
for all $i\in N$. Suppose not. Then there exists some $j\in N$ such that 
\[
\pi(z) \cdot \delta_j^\varepsilon + \rho(j,z) \pi(z) \cdot c(z) < \pi(z) \cdot e_j.
\]
As in {\bf Step I}, the function $\hat h:T\to \mathbb B_+$, defined by 
\[
\hat{h}(t): = \left\{
\begin{array}{ll}
\delta_i^\varepsilon-e_j, & \mbox{ if $t \in T_{j}$; }\\
-\rho(i,z)c(z), & \mbox{if $t \in T_i$, $i \neq j$,}
\end{array}
\right.
\]
satisfies 
\[
\pi(z) \cdot \left[\int_{T} \hat h\, d\mu + c(z)\right] < 0.
\]
This leads to a contradiction. Consequently, we have
\[
\pi(z) \cdot \delta_i^\varepsilon + \rho(i,z) \pi(z) \cdot c(z) \ge \pi(z) \cdot e_i
\]
for all $i\in N$. This implies that 
\[
\pi(z) \cdot \gamma_i^z + \rho(i,z) \pi(z) \cdot c(z) \ge \pi(z) \cdot e_i
\]
for all $i\in N$. Using the feasibility of $(\gamma_1^z,\cdots, \gamma_n^z, z)$, we conclude that
\[
\pi(z) \cdot \gamma_i^z + \rho(i,z) \pi(z) \cdot c(z) = \pi(z) \cdot e_i
\]
for all $i\in N$.

\medskip
\noindent
{\bf Step III:} Suppose that $\pi(z) \notin \mathbb B^*_{++}$. Then there exists some private good $\eta \in \mathbb B_+ \setminus \left\{0\right\}$ such that $\pi(z) \cdot \eta = 0$. 
Define $\zeta: = \gamma_{i_0}^z + \eta$. It then follows that $u_{i_0}(\zeta, z) > u_{i_0}(x_{i_0}, y)$ and $\pi(z) \cdot \zeta + \rho(i_0, z) \pi(z) \cdot c(z) = \pi(z) \cdot e_{i_0}$. This contradicts {\bf Step I}. Therefore, we conclude that $\pi(z)\in 
\mathbb B^*_{++}$.

\medskip
\noindent
{\bf Claim 3.} \emph{$(f(t), y)$ satisfies budget feasibility for all $t\in T$}. By {\bf Step II} of {\bf Claim 2}, as $\gamma_i^y=x_i$ for all $i\in N$, 
we note that $(f(t), y)$ is budget feasible for all $t\in T$.  

\medskip
\noindent
{\bf Claim 4.} \emph{$(f(t), y)$ is optimizing under the price system $\pi$ for all $t\in T$}. To this end, pick some $j\in N$, and choose 
some $(\zeta, z)\in \mathbb B_+\times \mathscr Y$ such that $u_j(\zeta, z) > u_j(x_j, y)$. We must show that
\[
\pi(z) \cdot \zeta + \rho(j,z)\pi(z) \cdot c(z) > \pi(z) \cdot e_j.
\] 
Suppose not, then $\pi(z) \cdot \zeta + \rho(j,z)\pi(z) \cdot c(z) \le \pi(z) \cdot e_j$. By the essentiality condition, it follows that $\zeta\neq 0$. Hence, there is some $\beta \in (0,1)$ such that $u_j(\beta\zeta, z) > u_j(x_j, z)$ and 
\[
\pi(z) \cdot \beta \zeta + \rho(j,z) \pi(z) \cdot c(z) < \pi(z) \cdot e_j.
\] 
Then, again as in {\bf Step I}, the function $\hat h:T\to \mathbb B_+$, defined by 
\[
\hat{h}(t): = \left\{
\begin{array}{ll}
\beta\zeta-e_j, & \mbox{ if $t \in T_{j}$; }\\
-\rho(i,z)c(z), & \mbox{if $t \in T_i$, $i \neq j$,}
\end{array}
\right.
\]
satisfies 
\[
\pi(z) \cdot \left[\int_{T} \hat h\, d\mu + c(z)\right] < 0.
\]
 This is a contradiction. Thus, our claim is verified.  This completes the proof of the theorem.
\end{proof}

Next, we extend Theorem \ref{thm:equivalence} to an economy with public goods and the equal treatment property whose
  commodity space is a Banach lattice whose positive cone has no interior points but possesses a quasi-interior point. For each 
$i\in N$ and each $(x_i, y)\in \mathbb B_+\times \mathscr Y$, let 
\[
P_i(x_i, y):=\{(\xi_i,z)\in \mathbb B_+\times \mathscr Y: u_i(\xi_i,z)> u_i(x_i,y)\}
\]
  be the set of all consumption bundles preferred to $(x_i,y)$ by agent $i$. Then 
$P_i: \mathbb B_+\times \mathscr Y\rightrightarrows \mathbb B_+\times \mathscr Y$ 
is called the \emph{preference relation} of agent $i$. For each $i\in N$, $(x_i, y)
  \in \mathbb B_+\times \mathscr Y$ and $z\in \mathscr Y$, we define 
$P_i^z: \mathbb B_+\times \mathscr Y\rightrightarrows \mathbb B_+$ by letting 
\[
P_i^z(x_i,y):=\{\xi_i\in \mathbb B_+:(\xi_i,z)\in P_i(x_i,y)\}.
\]
Note that $P_i^z(x_i,y)$
  is convex and relatively $\|\cdot\|$-open in $\mathbb B_+$ for
  all $i\in N$. Therefore, 
\[
P_i(x_i,y)= \bigcup\left\{P_i^z(x_i,y)\times \{z\}:z\in \mathscr Y\right\}.
\]
The following definition of ATY-properness is a modification of that in
  Podczeck and Yannelis (2008).

  \begin{definition} \label{defi:proper}
  The relation $P_i: \mathbb B_+\times \mathscr Y\rightrightarrows \mathbb B_+\times
 \mathscr Y$ is called \emph{ATY-proper} at $(x_i,y)\in \mathbb B_+\times \mathscr Y$ if 
for each $z\in \mathscr Y$ there exists a convex subset $\widetilde{P_i^z}
  (x_i,y)$ of $\mathbb B$ with non-empty $\|\cdot\|$-interior
  such that $\widetilde{P_i^z}(x_i,y)\cap \mathbb B_+= P_i^z(x_i,y)$ and
  ($\|\cdot\|$-int$\widetilde{P_i^z}(x_i,y)$)$\cap \mathbb B_+\neq
  \emptyset$.
  \end{definition}

  {\bf Properness:}. If $(x_1,\cdots,x_n, y)$ is a Pareto optimal allocation in $\mathscr{E}$, 
then for each $i\in N$, $P_i$ is ATY-proper at $(x_i,y)$.

  \begin{lemma}\label{lem:contfunc}
  Let $Y$ be a real vector space endowed with a Hausdorff, locally
  convex topology $\tau$ and let $U, V$ be convex subsets of
  $Y$ such that $U$ is open and $U\cap V\neq \emptyset$. Let $y\in
  V\cap {\rm cl} U$, where ${\rm cl}U$ denotes the closure of $U$.
  Suppose that $\pi$ is a linear functional (not necessarily
  continuous) on $Y$ with $\langle\pi, y\rangle\leq \langle\pi,
  y^\prime \rangle$ for all $y^\prime\in U\cap V$. Then, there exist
  linear functionals $\pi_1$ and $\pi_2$ on $Y$ such that $\pi_1$
  is continuous, $\langle \pi_1, y\rangle\leq \langle\pi_1, u
  \rangle$ for all $u\in U$, $\langle \pi_2, y\rangle\leq \langle
  \pi_2, v\rangle$ for all $v\in V$ and $\pi= \pi_1+ \pi_2$.
  \end{lemma}

  \begin{lemma}\label{lem:dense}
  Let $Y$ be a Riesz space endowed with a Hausdorff, locally
  convex topology $\tau$. Suppose that $L(z)$ denotes the 
principal ideal generated by $z$ in $Y$. If $L(z)$ is $\tau$-dense in $Y$,
  then $L(z)_+$ is $\tau$-dense in $Y_+$.
  \end{lemma}

  \begin{lemma}\label{lem:rieszdec}
  Let $Y$ be a Riesz space and let $Z$ be an ideal in $Y$. Let
  $y_1,..., y_m$ be elements of $Y$ and $z_1, ..., z_m$ be
  elements of $Z$ such that $\sum_{i= 1}^m z_i\leq \sum_{i= 1}^m
  y_i$. Suppose that there exists an element $z\in Z$ such that
  $z\leq y_i$ for each $i= 1, ..., m$. Then, there are elements
  $\hat{z}_1,\cdots, \hat{z}_m$ of $Z$ such that $\sum_{i= 1}^m
  \hat{z}_i= \sum_{i= 1}^m z_i$ and $\hat{z}_i\leq y_i$ for
  each $i= 1,\cdots, m$.
  \end{lemma}

  For proofs of Lemmas \ref{lem:contfunc}, \ref{lem:dense} and
  \ref{lem:rieszdec}, see Lemmas 2 and 3 in Podczeck (1996)
  and Lemma 7 in Podczeck and Yannelis (2008), respectively.
  In the proof of the next theorem,the method employed to establish the 
continuity of the equilibrium price closely follows the argument presented in Theorem 2 of Podczeck and Yannelis (2008).

  \begin{theorem}\label{thm:core-wal2}
  Let $\mathscr{E}_c$ be an equal treatment economy with public goods, constructed from a finite economy $\mathscr{E}$ such that $\mathbb B_+$
lacks an interior but has a quasi-interior point. Let $\sigma$ be a contribution measure for $\mathscr{E}_c$, and let $\widehat{\sigma}$ denote the corresponding contribution measure induced by $\sigma$ in $\mathscr{E}_c$.  Suppose $(f, y)$ is a feasible allocation in $\mathscr E_c$ such that $f(t) = x_i$ for all $t \in T_i$ and $i \in N$, and $(f, y)$ belongs to the $\widehat{\sigma}$-core of $\mathscr{E}_c$. Under the properness condition, there exists a non-zero price system $\pi: \mathscr{Y} \to \mathbb{B}^*_+$ such that $(f, y)$ is a cost share equilibrium of $\mathscr{E}_c$.
  \end{theorem}

  \begin{proof}
 Define $\mathbb C= L(\bf e)$ to be the principal ideal generated by ${\bf e}$ in $\mathbb B$, where 
${\bf e}=\sum_{i\in N}e_i$. Therefore, $(\mathbb C, \|\cdot\|_{\bf e})$ is an $AM$-space with $\bf e$
  as an order unit. Furthermore, 
 \begin{itemize}
\item[(i)] ${\bf e} \in \|\cdot\|_{\bf e}$-int$\mathbb C_+$;
\item[(ii)] $\mathbb C_+$ is $\|\cdot\|_{\bf e}$-closed in $\mathbb C$; and 
\item[(iii)] the $\|\cdot\|_{\bf e}$-closed unit ball of $\mathbb C$ coincides with
  the order interval $[-\bf{e}, \bf{e}]$.
\end{itemize}
  Define a new economy $\widetilde{\mathscr{E}}$ which is identical with $\mathscr{E}$
  except for the commodity space being $\mathbb C$ equipped with the
  $\|\cdot\|_{\bf{e}}$-topology, and
  agent $i$'s utility being $u_i(0,z)|_{\mathbb C}$ for each $z\in \mathscr Y$.
  If $(\xi_1,\cdots, \xi_n, z)$ is a feasible allocation of $\mathscr{E}$,
  then $\xi_i\in \mathbb C_+$ for each $i\in N$, and therefore, we have $x_i \in \mathbb C_+$ for each $i\in N$. 
  Since $(\mathbb C, \|\cdot\|_{\bf{e}})$ is a Banach lattice, the
  $\|\cdot\|$-topology is weaker than the
  $\|\cdot\|_{\bf{e}}$-topology on $\mathbb C$. Thus, 
  $u_i(\cdot, z)|_{\mathbb C}$ is $\|\cdot\|_{\bf{e}}$-continuous for each $z\in \mathscr Y$.
  Consequently, we have verified that $\widetilde{\mathscr{E}}$ satisfies
  hypothesis of Theorem \ref{thm:equivalence}, and $f$ is in the $\widehat{\sigma}$-core of
  $\widetilde{\mathscr{E}}_c$. By Theorem \ref{thm:equivalence}, there
  is a non-zero positive function $\widetilde{\pi}:\mathscr Y\to (\mathbb C, \|\cdot
  \|_{\bf{e}})^\ast$ such that $(f, y)$ is a cost share equilibrium of
  $\widetilde{\mathscr{E}}_c$ under the price system $ \widetilde{\pi}$. In what follows, we show that there is a non-zero
  positive function $\pi:\mathscr Y\to (\mathbb C, \|\cdot\|)^\ast$ such that $(f,y)$
 is a cost share equilibrium of $\mathscr{E}_c|_{\mathbb C}$ under the price system $\pi$, where $\mathscr{E}_c|_{\mathbb C}$ 
is identical with $\widetilde{\mathscr{E}}_c$ except for the commodity space being $\mathbb C$
  with the norm $\|\cdot\|$.

\medskip
  Since $(f,y)$ is a cost share equilibrium of $\widetilde{\mathscr{E}}_c$ with $f(t)=x_i$ for all $t\in T_i$ and $i\in N$, we 
have that $(x_1, \cdots,x_n, y)$ is a cost share equilibrium of
  $\widetilde{\mathscr{E}}$. Thus, $(x_1, \cdots,x_n, y)$ is a 
  Pareto optimal allocation in $\widetilde{\mathscr{E}}$, and also in
  $\mathscr{E}$. Let $i\in N$ and $z\in \mathscr Y$. By properness condition and Definition
  \ref{defi:proper}, there is a convex and
  $\|\cdot\|$-open subset $V_i$ of $\mathbb B$ such that
  \[
  \emptyset\neq V_i\cap \mathbb B_+ \subseteq P_i^z(x_i,y) \mbox{ and }
  \|\cdot\|-{\rm cl} P_i^z(x_i,y)\subseteq
  \|\cdot\|-{\rm cl} V_i.
 \] 
Since ${\bf e}$ is a quasi-interior point of $\mathbb B_+$,
  $\mathbb C$ is $\|\cdot\|$-dense in $\mathbb B$. By Lemma \ref{lem:dense},
  $\mathbb C_+$ is $\|\cdot\|$-dense in $\mathbb B_+$. Thus, $V_i\cap\mathbb C_+\neq 
\emptyset$. Define $W_i= V_i\cap \mathbb C$. Consequently, $W_i$ is convex and relatively
  $\|\cdot\|$-open in $\mathbb C$. Furthermore,  
\[
\emptyset\neq W_i\cap \mathbb C_+\subseteq P_i^z(x_i,y) \mbox{ and }
 \|\cdot\|_{\mathbb C}-{\rm cl}\widehat{P_i^z}(x_i,y)\subseteq \|\cdot
  \|_{\mathbb C}-{\rm cl}W_i,
\]
 where $\widehat{P_i^z}(x_i,y)= P_i^z(x_i,y)\cap \mathbb C_+$. By Remark \ref{rem:gamma}, 
there exists a feasible allocation $(\gamma_1^z, \cdots, \gamma_n^z, z)$ such that for all $i \in N$,
\begin{itemize}
\item[(iv)] $\gamma_i^z\in \|\cdot\|_{\mathbb C}$-${\rm cl}\widehat{P_i^z}(x_i,y)$; and 
\item[(v)] $\widetilde{\pi}(z)\cdot \gamma^z_i+ \rho(i,z)\widetilde{\pi}(z)\cdot c(z)= \widetilde{\pi}(z)\cdot e_i$.
\end{itemize}
 For any $\xi_i\in W_i\cap \mathbb C_+$, since $\xi_i\in P_i^z(x_i,y)$ and $(x_1, \cdots,x_n, y)$ is a cost share equilibrium 
under the price system $\widetilde{\pi}$, by (v), we have 
\begin{eqnarray}\label{eq:1}
\widetilde{\pi}(z)\cdot \xi_i+ \rho(i, z)\widetilde{\pi}(z)\cdot c(z)> \widetilde{\pi}(z)\cdot \gamma^z_i+ \rho(i, z)\widetilde{\pi}(z)\cdot c(z).
 \end{eqnarray}
From (iv), it follows that $\gamma_i^z\in \|\cdot\|_{\mathbb C}$-${\rm cl}W_i$. Define $\widehat{W_i}=W_i+\{\rho(i, z)c(z)\}$ and note that 
$\|\cdot\|_{\mathbb C}$-${\rm cl}{\widehat W_i}=\|\cdot\|_{\mathbb C}$-${\rm cl}W_i+\{\rho(i, z)c(z)\}$.Therefore, 
\begin{eqnarray}\label{eq:2}
\gamma_i^z+\rho(i, z)c(z)\in \mathbb C_+\cap \|\cdot\|_{\mathbb C}-{\rm cl}\widehat{W_i},
\end{eqnarray}
 By Lemma \ref{lem:contfunc}, Equations (\ref{eq:1}) and (\ref{eq:2}), there exist a $\pi_i^1(z)\in (\mathbb C,
  \|\cdot\|)^\ast$ and a linear functional $\pi_i^2(z)$ on
  $(\mathbb C, \|\cdot\|)$ such that 
\begin{itemize}
\item[(vi)] $\pi_i^1 (z)\cdot (\xi_i+ \rho(i, z)c(z))\ge \pi_i^1(z)\cdot (\gamma^z_i+\rho(i, z)c(z))$ for all
  $\xi_i\in W_i$; \footnote{Note that $\xi_i\in W_i$ if and only if $\xi_i+\rho(i, z)c(z)\in \widehat{W_i}$.}
\item[(vii)] $\pi_i^2 (z)\cdot (\xi_i+\rho(i, z)c(z))\ge \pi_i^2(z)\cdot (\gamma^z_i+\rho(i, z)c(z))$ for all
  $\xi_i\in \mathbb C_+$; \footnote{Note that $\xi_i\in \mathbb C_+$ if and only if $\xi_i+\rho(i, z)c(z)\in \mathbb C_+$.} and  
\item[(viii)] $\widetilde{\pi}(z)= \pi_i^1(z)+\pi_i^2(z)$ for all $z\in \mathscr Y$. 
\end{itemize}
Since $\mathbb C_+$ is a cone, we have $\pi_i^2(z)\cdot (\gamma^z_i+\rho(i, z)c(z))= 0$. It follows that  $\pi_i^2 (z)\cdot (\xi_i+ 
\rho(i, z)c(z))\ge 0$ for all
  $\xi_i\in \mathbb C_+$, which can equivalently be written as $\pi_i^2(z)\cdot \zeta_i\ge 0$ for all $\zeta_i\in 
  \mathbb C_+$. Consequently, we have
  \begin{itemize}
  \item[(ix)] $\widetilde{\pi}(z)\cdot (\gamma^z_i+\rho(i, z)c(z))= \pi_i^1(z)\cdot (\gamma^z_i+\rho(i, z)c(z))$; and 
  \item[(x)] $\widetilde{\pi}(z)\cdot \zeta_i\ge
  \pi_i^1(z)\cdot \zeta_i$ for all $\zeta_i\in \mathbb C_+$.
  \end{itemize}
By (x), we conclude that $\widetilde{\pi}(z)\ge \pi_i^1(z)$ for all $z\in \mathscr Y$. 
  Since $(\mathbb C,\|\cdot\|)$ is a locally solid Riesz space, $(\mathbb C,
  \|\cdot\|)^\ast$ is an ideal in $(\mathbb C, \|\cdot\|_{\bf e})^\ast$.
Hence, we can choose an element $\widehat{\pi}:\mathscr Y\to 
  (\mathbb C, \|\cdot\|)^\ast$ such that $\widehat{\pi} (z)=
  \sup\{\pi^1_i(z): i\in N\}$ for all $z\in \mathscr Y$. Therefore, $\widehat{\pi}(z)
  \in (\mathbb C, \|\cdot\|)^\ast$ and $\widehat{\pi}(z)\le
  \widetilde{\pi}(z)$ for all $z\in \mathscr Y$. By the Riesz-Kantorovich formulas and (ix), we
  obtain
  \begin{eqnarray*}
  \widehat{\pi}(z)\cdot {\bf e}
  &=& \sup\left\{\sum_{i\in N}\pi^1_i(z)\cdot \zeta_i: \zeta_i\in \mathbb C_+, 
\sum_{i\in N}\zeta_i= {\bf e}\right\} \nonumber\\
  &\geq& \sum_{i\in N}\pi^1_i(z)\cdot (\gamma_i^z+\rho(i, z)c(z))\noindent \\
  &=& \sum_{i\in N}\widetilde{\pi}(z) \cdot (\gamma_i^z+\rho(i, z)c(z)) \noindent \\
 &=& \widetilde{\pi}(z)\cdot \sum_{i\in N}(\gamma_i^z+\rho(i, z)c(z)) \noindent \\
  &=&  \widetilde{\pi}(z)\cdot {\bf e}.
  \end{eqnarray*}
   Choose an element $\xi\in \mathbb C_+$. 
Since $\mathbb C=L(\bf e)$, there is some $\delta> 0$ such
  that $\xi\leq \delta {\bf e}$. This along with the fact that $\widehat{\pi}(z)\leq \widetilde{\pi}(z)$
imply that 
\[
\widehat{\pi}(z)\cdot (\delta{\bf e}-\xi)\leq \widetilde{\pi}(z)\cdot (\delta{\bf e}-\xi).
\]
Consequently, we have $\widehat{\pi}(z)\cdot \xi\geq \widetilde{\pi}(z)\cdot \xi$. 
Hence, $\widehat{\pi}(z)\geq \widetilde{\pi}(z)$ and therefore, $\widehat{\pi}(z)= \widetilde{\pi}(z)$
for all $z\in \mathscr Y$. Thus, $\widetilde{\pi}(z)\in (\mathbb C, \|\cdot\|)^\ast$ for 
all $z\in \mathscr Y$ and $(f, y)$ is a cost share equilibrium
of $\mathscr{E}_c|_{\mathbb C}$. By the Hahn-Banach
  theorem, we can choose a positive element $\ddot{\pi}(z)\in (\mathbb B, \|\cdot
  \|)^\ast$ for each $z\in \mathscr Y$ such that $\ddot{\pi}(z)$ is an extension of $\widetilde{\pi}(z)$.
  Since $\mathbb C_+$ is $\|\cdot\|$-dense in $\mathbb B_+$
  and $u_i(\cdot,z)$ is $\|\cdot\|$-continuous for each $i\in N$,
  we deduce that 
\begin{eqnarray}\label{eqn:cost}
\ddot{\pi}(z)\cdot \xi+ \rho(i, z)\ddot{\pi}(z)\cdot c(z)\ge \ddot{\pi}(z)\cdot e_i
\end{eqnarray}
for all $\xi\in \mathbb B_+$ satisfying $u_i(\xi,z)> u_i(x_i, y)$. Analogous to {\bf Claim 4} in the proof of Theorem \ref{thm:equivalence}, 
one can show that $(f,y)$ is a cost share equilibrium under the price 
system $\ddot{\pi}:\mathscr Y\to \mathbb B^*$.   
  \end{proof}

  Now, we extend Theorem \ref{thm:core-wal2} to an economy with the equal 
treatment property whose commodity space is a Banach lattice having no quasi-interior
  point in its positive cone. The following definition of strong
  ATY-properness and the argument to get continuity of the
  equilibrium price in the next theorem are similar to those
  in (A8$^\prime$) and Theorem 3 of Podczeck and Yannelis (2008).

  \begin{definition} \label{defi:strongproper}
 The relation $P_i: \mathbb B_+\times \mathscr Y\rightrightarrows \mathbb B_+\times
 \mathscr Y$ is called \emph{strongly ATY-proper} at $(x_i,y)\in \mathbb B_+\times \mathscr Y$ if 
for each $z\in \mathscr Y$ there exists a convex subset $\widetilde{P_i^z}
  (x_i,y)$ of $\mathbb B$ with non-empty $\|\cdot\|$-interior
  such that 
\[
\widetilde{P_i^z}(x_i,y)\cap \mathbb B_+= P_i^z(x_i,y) \mbox{ and }
  (\|\cdot\|-{\rm int} \widetilde{P_i^z}(x_i,y))\cap \mathbb B_+\cap L\left(\sum_{i\in N} e_i\right)\neq
  \emptyset,
\]
where $L(\sum_{i\in N}e_i)$ to be the principal ideal generated by $\sum_{i\in N}e_i$ in $\mathbb B$.
  \end{definition}

 {\bf Strong properness:} If $(x_1, \cdots,x_n, y)$ is a Pareto optimal allocation of $\mathscr{E}$, 
then for each $i\in N$, $P_i$ is strongly ATY-proper at $(x_i,y)$.

  \begin{theorem}\label{thm:core-wal3}
Let $\mathscr{E}_c$ be an equal treatment economy with public goods, constructed from a finite economy $\mathscr{E}$ such that $\mathbb B_+$
has no quasi-interior point. Let $\sigma$ be a contribution measure for $\mathscr{E}$, and let $\widehat{\sigma}$ denote the corresponding contribution measure induced by $\sigma$ in 
$\mathscr{E}_c$.  Suppose $(f, y)$ is a feasible allocation in $\mathscr E_c$ such that $f(t) = x_i$ for all $t \in T_i$ and $i \in N$, and $(f, y)$ belongs to the $\widehat{\sigma}$-core of $\mathscr{E}_c$. Under the strong properness condition, there exists a non-zero price system $\pi: \mathscr{Y} \to \mathbb{B}^*_+$ such that $(f, y)$ is a cost share equilibrium of $\mathscr{E}_c$.
 \end{theorem}

  \begin{proof}
 Define $\mathbb C= L(\bf e)$ to be the principal ideal generated by ${\bf e}$ in $\mathbb B$, where 
${\bf e}=\sum_{i\in N}e_i$.
  Then, $(\mathbb D,\|\cdot\|)$ equipped with the ordering of
  $(\mathbb B, \|\cdot\|)$ is a Banach lattice, where $\mathbb D$ denotes the
  $\|\cdot\|$-closure of $\mathbb C$ in $\mathbb B$.  Note that if $(\xi_1,\cdots,\xi_n, z)$ is a feasible 
 allocation of $\mathscr{E}$, then $\xi_i\in \mathbb C_+$ for each $i\in N$, and therefore, 
we have $x_i \in \mathbb C_+$ for each $i\in N$. Clearly, agent $i$'s utility function 
$u_i(0,z)|_{\mathbb D}$ for each $z\in \mathscr Y$ satisfies continuity, monotonicity and quasi-concavity assumptions.
 Suppose that $(h_1,\cdots, h_n, w)$ is a Pareto optimal allocation of the economy
  $\mathscr{E}|_{\mathbb D}$, which is identical with $\mathscr{E}$ except
  for the commodity space being $\mathbb D$, each agent's private consumption set
  being $\mathbb D_+$, and agent $i$'s utility being $u_i(0,z)|_{\mathbb D}$ for each $z\in 
\mathscr Y$. Then $(h_1, \cdots, h_n, w)$ is Pareto optimal allocation of $\mathscr{E}$. Take $\widehat{P_i^z}
  (h_i,w)= \widetilde{P_i^z}(h_i,w)\cap \mathbb D$ for each $i\in N$ and $z\in \mathscr Y$, where
  $\widetilde{P_i^z}(h_i,w)$ is chosen according to strong properness condition and Definition
  ~\ref{defi:strongproper}. Therefore, for each
  $i\in N$, $\widehat{P_i^z}(h_i,w)$ is convex with non-empty relative
  $\|\cdot\|$-interior in $\mathbb D$. By strong properness condition and Definition
  \ref{defi:strongproper}, for each $i\in N$,
  \[
 \widehat{P_i^z}(h_i,w)\cap \mathbb D_+= P_i^z(h_i,w)|_{\mathbb D} \mbox{ and }
 (\|\cdot\|-{\rm cl} \widehat{P_i^z}(h_i,w))\cap \mathbb D_+\neq \emptyset,
\]
 where $P_i^z(h_i,w)|_{\mathbb D}=P_i^z(h_i,w)\cap \mathbb D$. Therefore, the properness condition and 
the existence of a quasi-interior point in
$\mathbb D_+$ hold for the restricted economy
  $\mathscr{E}|_{\mathbb D}$. Note that $(f,y)$ is in the $\widehat{\sigma}$-core of
  $\mathscr{E}_c|_{\mathbb D}$. By Theorem \ref{thm:core-wal2}, there exists a
  non-zero positive function $\widehat{\pi}:\mathscr Y\to {\mathbb D}^*$ such that $(f,
  y)$ is a cost share equilibrium allocation of $\mathscr{E}_c|_ {\mathbb D}$ under the price system 
$\widehat{\pi}$. Therefore, 
  $(x_1, \cdots, x_n, y)$ is a cost share equilibrium allocation of $\mathscr{E}|_{\mathbb D}$ under the price system $\widehat{\pi}$. 
  By the Hahn-Banach theorem, there is a non-zero positive function  $\widetilde{\pi}:\mathscr Y\to {\mathbb B}^*$
  such that $\widetilde{\pi}(z)$ is an extension of $\widehat{\pi}(z)$ for all $z\in \mathscr Y$. Then $(x_1, \cdots,x_n,  y)$
  satisfies all conditions of cost share equilibrium of $\mathscr E$ except for the fact that 
\[
\widetilde{\pi}(z)\cdot \xi_i+ \rho(i, z)\widetilde{\pi}(z)\cdot c(z)> \widetilde{\pi}(z)\cdot e_i,
 \]
 for all $\xi_i\in \mathbb B_+ \setminus \mathbb D_+$ satisfying $u_i(\xi_i, z)> u_i(x_i,y)$ and 
$i\in N$.

\medskip
  Since $(x_1,\cdots,x_n, y)$ is a cost share equilibrium of $\mathscr{E}|_{\mathbb D}$ under the price 
 system $\widehat{\pi}$, $(x_1, \cdots, x_n,  y)$ is a Pareto optimal allocation in $\mathscr{E}|_{\mathbb D}$ and hence, in 
$\mathscr{E}$. Pick an $i\in N$. By the strong properness condition and Definition \ref{defi:strongproper}, there is a convex and
  $\|\cdot\|$-open subset $V_i$ of $\mathbb B$ such that
  \[
  \emptyset\neq V_i\cap \mathbb D_+ \subseteq P_i^z(x_i,y)|_{\mathbb D} \mbox{ and }
  \|\cdot\|-{\rm cl} P_i^z(x_i,y)\subseteq
  \|\cdot\|-{\rm cl} V_i.
 \] 
 By Remark \ref{rem:gamma}, 
there exists a feasible allocation $(\gamma_1^z, \cdots, \gamma_n^z,  z)$ such that for all $i\in N$,\footnote{Note that $\gamma_i^z\in \mathbb D$ for all $i\in N$.}
\begin{itemize}
\item[(i)] $\gamma_i^z+\rho(i, z)c(z)\in \mathbb D_+$ and $\gamma_i^z\in \|\cdot\|$-${\rm cl}P_i^z(x_i,y)|_{\mathbb D}$; and 
\item[(ii)] $\widetilde{\pi}(z)\cdot \gamma^z_i+ \rho(i, z)\widetilde{\pi}(z)\cdot c(z)= \widetilde{\pi}(z)\cdot e_i$.
\end{itemize}
 For any $\xi_i\in V_i\cap \mathbb D_+$, since $\xi_i\in P_i^z(x_i,y)$ and $(x_1,\cdots, x_n, y)$ is a cost share equilibrium allocation of $\mathscr E|_{\mathbb D}$
under the price system $\widetilde{\pi}$, by (ii), we have 
\begin{eqnarray}\label{eq:3.4}
\widetilde{\pi}(z)\cdot \xi_i+ \rho(i, z)\widetilde{\pi}(z)\cdot c(z)> \widetilde{\pi}(z)\cdot \gamma^z_i+ \rho(i, z)\widetilde{\pi}(z)\cdot c(z).
 \end{eqnarray}
From (ii), it follows that $\gamma_i^z\in \|\cdot\|$-${\rm cl}V_i$. Define $W_i:=V_i+\{\rho(i, z)c(z)\}$ and note that 
$\|\cdot\|$-${\rm cl} W_i=\|\cdot\|$-${\rm cl}V_i+\{\rho(i, z)c(z)\}$.Therefore, 
\begin{eqnarray}\label{eq:3.5}
\gamma_i^z+\rho(i, z)c(z)\in \mathbb D_+\cap \|\cdot\|-{\rm cl}W_i.
\end{eqnarray}
 By Lemma \ref{lem:contfunc}, Equation (\ref{eq:3.4}) and Equation (\ref{eq:3.5}), there exist a $\pi_i^1(z)\in \mathbb B^\ast$ and a linear functional $\pi_i^2(z)$ on
  $\mathbb B$ such that 
\begin{itemize}
\item[(iii)] $\pi_i^1 (z)\cdot (\xi_i+\rho(i, z)c(z))\ge \pi_i^1(z)\cdot (\gamma^z_i+\rho(i, z)c(z))$ for all
  $\xi_i\in V_i$; 
\item[(iv)] $\pi_i^2 (z)\cdot (\xi_i+\rho(i, z)c(z))\ge \pi_i^2(z)\cdot (\gamma^z_i+\rho(i, z)c(z))$ for all
  $\xi_i\in \mathbb D_+$; and  
\item[(v)] $\widetilde{\pi}(z)= \pi_i^1(z)+\pi_i^2(z)$ for all $z\in \mathscr Y$. 
\end{itemize}
Since $\mathbb D_+$ is a cone, we have $\pi_i^2(z)\cdot (\gamma^z_i+\rho(i, z)c(z))= 0$. It follows that  $\pi_i^2 (z)\cdot (\xi_i+ 
\rho(i, z)c(z))\ge 0$ for all
  $\xi_i\in \mathbb D_+$, which can equivalently be written as $\pi_i^2(z)\cdot \zeta_i\ge 0$ for all $\zeta_i\in 
  \mathbb D_+$. Consequently, we have
  \begin{itemize}
  \item[(vi)] $\widetilde{\pi}(z)\cdot (\gamma^z_i+\rho(i, z)c(z))= \pi_i^1(z)\cdot (\gamma^z_i+\rho(i, z)c(z))$; and 
  \item[(vii)] $\widetilde{\pi}(z)\cdot \zeta_i\ge
  \pi_i^1(z)\cdot \zeta_i$ for all $\zeta_i\in \mathbb D_+$.
  \end{itemize}
    Since $\pi_1^i(z)\in \mathbb B^*$ and $\|\cdot
  \|$-cl$P_i^z(x_i,y)\subseteq \|\cdot\|$-${\rm cl}V_i$, from (iii), it follows that 
  \[
  \pi_i^1 (z)\cdot (\xi_i+\rho(i, z)c(z))\ge \pi_i^1(z)\cdot (\gamma^z_i+\rho(i, z)c(z))
  \] 
  for all $\xi_i\in P_i^z(x_i,y)$. In view of (vi), we have 
  \begin{itemize}
    \item[(viii)] $\pi_i^1 (z)\cdot (\xi_i+\rho(i, z)c(z))\ge \widetilde{\pi}(z)\cdot (\gamma^z_i+\rho(i, z)c(z))$   for all $\xi_i\in P_i^z(x_i,y)$.
  \end{itemize}
  Since $\mathbb B$ is a locally solid Riesz space, $\mathbb B^\ast$ is an ideal
  in the order dual of $\mathbb B$. 
  Define $\ddot{\pi}:\mathscr Y\to \mathbb B^\ast$ such that $\ddot{\pi}(z)
  =\sup\{\pi_i^1(z): i\in N\}$ for each $z \in \mathscr Y$.
  By the Riesz-Kantorovich formulas and techniques similar to those in
  Theorem \ref{thm:core-wal2}, we have 
  \[
  \ddot{\pi}(z)\cdot {\bf e}\geq \sum_{i\in N} \pi_i^1 (z)(\gamma_i^z+\rho(i, z)c(z)).
  \]
   Using (vi) and the feasibility of $(\gamma_1^z, \cdots, \gamma_n^z, z)$, we can obtain that $\ddot{\pi}(z)\cdot {\bf e}\geq \widetilde{\pi}(z)\cdot {\bf e}$.
 Moreover, the Riesz-Kantorovich formulas
  and (vii) imply that $\widetilde{\pi}(z)\cdot \zeta\geq \ddot{\pi}(z)\cdot \zeta$ for all 
  $\zeta\in \mathbb D_+$. Since
  $\mathbb D= L({\bf e})$, we have $\ddot{\pi}(z) \equiv \widetilde{\pi}(z)$ on
  $\mathbb D$ for each $z\in \mathscr Y$, which can be combined with (ii) and (viii) 
  to derive 
  \[
  \ddot{\pi}(z)\cdot \xi_i+ \rho(i, z)\ddot{\pi}(z)\cdot c(z)\ge \ddot{\pi}(z)\cdot e_i
  \]
  for all $\xi_i\in P_i^z(x_i, y)$. Analogous to {\bf Claim 4} in the proof of Theorem \ref{thm:equivalence}, one can verify that $(f, y)$
  is a cost share equilibrium of
  $\mathscr{E}_c$. 
  \end{proof}

\subsection{The Size of Blocking Coalitions}\label{sec:vind}
In this subsection, we characterize the $\sigma$-core in terms of the size of blocking coalitions, thus generalizing the results of Schmeidler (1972) and Vind (1972) to 
our framework.
\begin{theorem}\label{thm:infintevind}
  Let $\mathscr E_c$ be an equal treatment economy with public projects constructed from the finite economy $\mathscr E$. 
  Let $\sigma$ be a contribution measure for $\mathscr{E}_c$, and let $\widehat{\sigma}$ denote the corresponding contribution
measure induced by $\sigma$ in $\mathscr{E}_c$. Let $(f,y)$ be a feasible allocation in
  $\mathscr{E}_c$ such that $(f(t), y)= (x_i, y)$ for all $t\in T_i$
  and $i\in N$. If $(f, y)$ is not in the $\widehat{\sigma}$-core of
  $\mathscr{E}_c$, then for any $0<\varepsilon<1$, there is a
  coalition $S$ with $\mu(S)=\varepsilon$ blocking $(f,y)$.
  \end{theorem}

\begin{proof}
Since $(f,y)$ is blocked in $\mathscr E_c$, there exists a coalition $S \subseteq T$ that blocks $(f,y)$ via an allocation $(g,z)$ such that

\begin{itemize}

\item[(i)] $\int_{S} g \, d\mu + \widetilde{\sigma}(S,z)c(z) = \int_{S} e\, d\mu$; and

\item[(ii)] $u_t(g(t),z) > u_t(f(t),y)$ $\mu$-a.e. on $S$.

\end{itemize}
Define $S_i = S \cap T_i$ for all $i \in N$ and $\widetilde{S}: = \left\{i \in N: \mu(S_i) > 0\right\}$. For each $i\in \widetilde{S}$, 
define 
\[
g_i = \frac{1}{\mu(S_i)}\int_{S_i}g\, d\mu.
\]  
Therefore, one can rewrite the condition (i) as
\begin{equation}\label{eqnv1}
\sum\limits_{i \in \widetilde{S}}\mu(S_i)g_i + \sum\limits_{i \in \widetilde{S}} \sigma(\left\{i\right\}, z)\frac{\mu(S_i)}{\mu(T_i)}c(z) = \sum\limits_{i \in \widetilde{S}}\mu(S_i) e_i.
\end{equation}
Choose an element $\delta \in (0,1)$. Since $\mu$ is atomless, there exists some $E_i\in \mathscr T$ such that $E_i \subseteq S_i$ and 
$\mu(E_i) = \delta\mu(S_i)$. From Equation (\ref{eqnv1}), it follows that 
\[
\sum\limits_{i \in \widetilde{S}}\mu(E_i)g_i + \sum\limits_{i \in \widetilde{S}} \sigma(\left\{i\right\},z)\frac{\mu(E_i)}{\mu(T_i)}c(z) = \sum\limits_{i \in \widetilde{S}}\mu(E_i) e_i.
\]
Let $E = \bigcup\{E_i:i\in \widetilde{S}\}$ and define $h:T\to \mathbb B_+$ by letting $h(t)=g_i$ if $t\in S_i$ and $i\in \widetilde{S}$; and 
$h(t)=x_i$, if $t\in T_i\setminus S_i$ and $i\in N$. Hence, one can observe from the above equation that
\[
\int_{E} h\, d\mu + \widetilde{\sigma}(E, z)c(z) = \int_{E} e\, d\mu, 
\]
where $\mu(E) =\delta\mu(S)$. Thus there exists a coalition $E \subseteq S$ with $\mu(E) = \delta\mu(S)$ that blocks $(f,y)$. This proves for $\varepsilon \leq \mu(S)$. 

\medskip
If $\mu(S) = 1$, then our proof is done. So suppose not and that $\mu(T \setminus S) > 0$.  Let $\varepsilon> 0$ be a number such that 
$\mu(S)< \varepsilon< \mu(T)$. Take $\delta\in (0, 1)$ such that 
\[
\delta= 1-\frac{\varepsilon-\mu(S)}{\mu(T\setminus S)}.
\]
By assumption of continuity, we can find some $\beta\in \mathbb B_+\setminus \{0\}$ and an allocation of  private goods $\widetilde{g}: S \to \mathbb B_+$ 
such that $u_t(\widetilde{g}(t), z) > u_t(f(t), y)$ $\mu$-a.e. on $S$ and
\begin{eqnarray}\label{eqn:tildeg}
\int_{S} \widetilde{g}\, d\mu = \int_{S} g\, d\mu - \beta.
\end{eqnarray}
By Remark \ref{rem:gamma}, there exists a feasible allocation $(\gamma_1^z,\cdots, \gamma_n^z, z)$ in $\mathscr E$ such that 
$u_i(\gamma_i^z,z) \ge u_i(x_i,y)$ for all 
$i\in N$. For each $z\in \mathscr Y$, define $\varphi^z:T\to \mathbb B_+$ by letting $\varphi^z(t)=\gamma^z_i$ for all $t\in T_i$ and all $i\in N$. Then 
$(\varphi^z, z)$ is a feasible allocation in $\mathscr E_c$ satisfying $u_t(\varphi^z(t),z) \ge u_t(f(t),y)$ for all $t\in T$. 
Define $\widetilde{g}_{\delta}: S \to \mathbb B_+$ as 
\[
\widetilde{g}_{\delta}(t) = \delta \widetilde{g}(t) + ( 1- \delta)\varphi^{z}(t).
\] 
By the quasi-concavity of $u_t$, we have $u_t(\widetilde{g}_{\delta}(t), z) > u_t(f(t),y)$ $\mu$-a.e. on $S$. Let $C = T \setminus S$ and 
define $C_i = C \cap T_i$ for all $i\in N$.  Let $\widetilde{C} = \left\{i \in N: \mu(C_i) >0\right\}$. Since $\mu$ is atomless, one can choose 
an element $D_i\in \mathscr T$ such that $D_i \subseteq C_i$ and $\mu(D_i) = (1-\delta)\mu(C_i)$ for all $i\in \widetilde C$. Hence, using 
the definition of $\widetilde{\sigma}$ along with the facts that $\varphi^z(t)= \gamma_i^z$ and $e(t)=e_i$ for all $t\in T_i$ and $i\in N$, we get  
\[
\int_{D_i} \varphi^z\, d\mu + \widetilde{\sigma}(D_i,z)c(z) - \int_{D_i}e\, d\mu = (1 - \delta)\left[\int_{C_i} \varphi^z\, d\mu + \widetilde{\sigma}(C_i, z)c(z) - \int_{C_i}e\, d\mu\right].
\]
Let $D= \bigcup\{D_i:i\in \widetilde C\}$. Then $\mu(D) = (1 - \delta)\mu(C)$ and one has
\[
\int_{D} \varphi^z\, d\mu + \widetilde{\sigma}(D,z)c(z) - \int_{D}e\, d\mu = (1 - \delta)\left[\int_{C} \varphi^z\, d\mu + \widetilde{\sigma}(C, z)c(z) - \int_{C}e\, d\mu\right].
\]
Define $h: D \to \mathbb B_+$ as
\[
h(t) = \varphi^{z}(t) + \frac{\delta\mu(S)}{\mu(D)} \beta.
\]
Let $E = S \cup D$. Then $\mu(E) = \mu(S) + (1-\delta)\mu(T \setminus S)=\varepsilon$. Define an allocation of private goods $\psi: T \to 
\mathbb B_+$ as $\psi(t): = \widetilde{g}_{\varepsilon}(t)$ if $t \in S$ and $\psi(t): =h(t)$, otherwise. Consequently, $u_t(\psi(t), z) > 
u_t(f(t), y)$ $\mu$-a.e. on $E$. Using the definitions of $\widetilde{g}_\delta$, $h$ and Equation (\ref{eqn:tildeg}), we have  
\begin{eqnarray*}
\int_{E} (\psi-e)\, d\mu + \widetilde{\sigma}(E,z) c(z) 
&= &\left[\int_{S}(\widetilde{g}_{\delta}-e)\, d\mu + \widetilde{\sigma}(S,z)c(z) \right] + \left[ \int_{D}(h-e)\, d\mu + \widetilde{\sigma}(D,z)c(z) \right]\noindent \\
&= & \delta\left[ \int_{S} (g-e)\, d\mu + \widetilde{\sigma}(S,z)c(z) \right] + (1 - \delta)\left[\int_{T}(\varphi^z-e)\, d\mu + \widetilde{\sigma}(T,z)c(z)\right]\noindent \\
&=& 0
\end{eqnarray*}
This completes the proof. 
\end{proof}

\section{Equivalence in Finite Economies}\label{sec5}
This section builds on the results of Subsections \ref{sec:equi} and \ref{sec:vind} to establish a characterization of cost share equilibria in a finite economy endowed with an abstract set of public projects and an infinite dimentional private commodity space. In particular, we provide several characterizations of cost share equilibria, formulated in terms of Aubin $\sigma$-core allocations, 
$\sigma$-Edgeworth equilibria, and non-dominated allocations. For the remainder of the paper, the following assumption will be maintained implicitly, with all subsequent results derived under one of the three alternative conditions specified below:
\begin{itemize}
\item[(i)] The positive cone $\mathbb{B}_+$ possesses a nonempty interior.
\item[(ii)] The positive cone $\mathbb{B}_+$ has an empty interior but admits a quasi-interior point, together with the satisfaction of the properness condition.
\item[(iii)] The positive cone  $\mathbb{B}_+$ admits no quasi-interior point, and the strong properness condition holds.
\end{itemize}
\subsection{Fuzzy Coalitions}
To begin with, we introduce the notion of a coalition in our finite economy $\mathscr E$, but in the sense of Aubin. 

\begin{definition}\label{defn:Aubin}
Define the set
\[
\mathcal A_c = \left\{ \gamma = (\gamma_1,\cdots,\gamma_n) \in [0,1]^n : \gamma_i > 0 \mbox{ for atleast one } i \in N\right\} 
\]
We call an element $\gamma \in \mathcal A_c$ as an {\it Aubin} (or {\it generalized}) {\it coalition}. Furthermore, the set 
$\left\{i \in N: \gamma_i > 0\right\}$ is called the {\it support} of the Aubin coalition $\gamma$ and is denoted by ${\rm supp}\gamma$. 
\end{definition}
The set $\mathcal A_c$ may be interpreted as the collection of generalized coalitions, in the sense that, for any $i \in N$, $\gamma_i$ represents the share of endowment employed by agent $i$ in the coalition $\gamma$. It is evident that the set of ordinary coalitions $\Sigma$ is contained in $\mathcal A_c$ as each ordinary coalition $S$ can be identified with its characteristic function $\chi_S$. As a subsequent step, we extend the notion of the veto mechanism from $\Sigma$ to the broader class of Aubin coalitions by enlarging each contribution scheme $\sigma$ from ordinary coalitions to generalized ones. Given that $\sigma(S,z) = \sum_{i\in S}\sigma(\{i\},z)$ for each coalition $S \in \Sigma$ and each $z\in \mathscr Y$, the contribution of the Aubin coalition $\gamma$ to the realization of a public project 
$z$ is given by
\[ 
\widetilde{\sigma}(\gamma, z) = \sum\limits_{i \in N} \gamma_i\sigma(\{i\}, z).
\]
In particular, when $\gamma_i=1$ if $i\in S$; and $\gamma_i=0$, otherwise, then it follows that $\widetilde{\sigma}(\gamma, z) = \sigma(S, z)$, and thus $\widetilde{\sigma}$ can be thought as a generalization of the contribution scheme $\sigma$. Let $\rho$ denote the individual cost distribution function corresponding to $\sigma$. Therefore, one can equivalently write 
$\widetilde{\sigma}$ as
\[
\widetilde{\sigma}(\gamma, z) = \sum\limits_{i\in N} \gamma_i\rho(i, z)
\] 
for each $\gamma\in \mathcal A_c$ and each $z\in \mathscr Y$. 
It follows that the individual cost contribution of agent $i \in N$ when participating in the Aubin coalition $\gamma$ is equal to $\gamma_i\rho(i,z)$, for each $i \in {\rm supp}\gamma$. The subsequent theorem investigates the relationship between the Aubin $\sigma$-core and cost share equilibria in a finite economy with public projects.

\begin{theorem}\label{fuzzyequivalence}
Let $\mathscr E$ be a finite economy. Suppose that $\sigma$ is a contribution measure with the corresponding cost distribution function $\rho$. Then $\mathscr{C}^A_{\sigma}(\mathscr E) = CE_{\rho}(\mathscr E)$.
\end{theorem}

\begin{proof}
The proof is analogous to the proof of Theorem 4.1 in Graziano and Romaniello (2012) and is relegated to Appendix. 
\end{proof}

A possible interpretation of the above equivalence theorem lies in the classical convergence results in terms of replicating a finite economy. On that note, define for each positive integer $r$ , the $r$-fold replica economy of $\mathscr E$ as the economy $\mathscr E_r$. The construction of the replicated economy entails replacing each agent $i$ of economic weight one with $r$ identical agents of total economic weight equal to $r$. The replicated economy has the following features:

\begin{itemize}

\item the economy $\mathscr E_r$ has the same commodity-price duality of $\mathscr E$; the same set of public projects $\mathscr Y$; the cost function defined by $c_r(z) = rc(z)$;

\medskip
\item for each $i =1,\cdots,n$, there are $r$ agents of type $i$, each one is indexed by $(i,j)$ with $j =1,\cdots,r$, having the same initial endowment $e_{(i,j)} = e_i$ and the same utility functions $u_{(i,j)}(\cdot,z) = u_i(\cdot,z)$ for any $z \in \mathscr Y$.

\end{itemize}
In an economy devoid of any public projects, the total initial endowment of the coalition made by the $r$-many agents of type $i$  becomes $re_i$, and then the total initial endowment of the replicated economy is given by $r\sum_{i\in N}e_i$. In the presence of public project, assuming that the cost of public projects are being paid upfront, one can argue that $e_i - \sigma(\{i\},z)c(z)$ is the bundle of goods that agent $i$ can exchange in the market. Consequently, in the corresponding $r$-fold replica economy $\mathscr E_r$, the $r$-many agents of type $i$ can exchange the commodity bundle $r(e_i - \sigma(\{i\},z)c(z))$. Summing over all agents in $\mathscr E_r$, the total commodity bundle available for exchange in the market is given by $r(\sum\limits_{i \in N}e_i - c(z))$. This, in comparison to $r\sum\limits_{i \in N} e_i$, yields that a natural choice for a cost function in the replicated economy is $rc(z)$. In economies characterized by an increasing number of agents, the specification of such a cost function resonates with the fact that there is an increasing level of cost for supplying the provision of public projects to a larger population. Within this framework, $z \in \mathscr Y$ is naturally interpreted as the provision level of a public good, exemplified by goods such as defense or security.

\begin{definition}
Let $\sigma \in \mathscr M$ be a contribution measure of the economy $\mathscr{E}.$ A contribution measure of the $r-$fold replica economy $\mathscr{E}_{r}$ is defined as $\sigma(\{i,j\})=\frac{\sigma(\{i\})}{r},$ for each $j=1,\cdots,r.$ A feasible allocation $(x_{1},\cdots,x_{n},y)$ of the economy $\mathscr E$ is an $\sigma$-{\it Edgeworth equilibrium} of 
$\mathscr{E}$ whenever the corresponding equal treatment allocation
\[
(x_{(1,1)},\cdots,x_{(1,r)},\cdots,x_{(n,1)},\cdots,x_{(n,r)}, y)
\]
with $x_{(i,h)}=x_{(i,k)}$ for any $h,k=1,\cdots,r$ and any $i=1,\cdots,n$, belongs to the $\sigma$-{\it core} of $\mathscr{E}_{r}$, for each $r.$ 
\end{definition}

\noindent
In the following proposition, we interpret the Aubin  $\sigma$-core  allocations of the finite economy $\mathscr E$ in terms of the  
$\sigma$-Edgeworth equilibria. This result, combined with Theorem \ref{fuzzyequivalence}, yields that cost share equilibria is equivalent to the $\sigma$-Edgeworth equilibria.

\begin{proposition}\label{eeequivalence}
Assume that $\sigma$ is a contribution measure in $\mathscr E$, for which $e_{i} - \sigma(\{i\},z)c(z) \gg 0$ holds  for each agent $i \in N$, and $z \in \mathscr{Y}.$ Then the $\sigma$-Aubin core of the economy  $\mathscr{E}$ coincides with the $\sigma$-Edgeworth equilibria.
\end{proposition}

\begin{proof}
The proof is analogous to the proof of Proposition 4.3 in Graziano and Romaniello (2012) and is relegated to Appendix. 
\end{proof}

\subsection{Non-dominated allocations}   
Consider a feasible allocation $(\xi_1, \ldots, \xi_n, z)$ within the finite economy $\mathscr{E}$, together with a vector of real numbers 
$\alpha = (\alpha_1, \ldots, \alpha_n)$ with $0\le \alpha_i\le 1$. We define by $\mathscr{E}(\xi_1,\cdots,\xi_n, z,\alpha)$ the auxiliary economy that coincides with $\mathscr{E}$ except for 
the initial endowment of each agent $i$, which is specified as
\[
e(z,\alpha)_i \;=\; \alpha_i e_i + (1-\alpha_i)\big[\xi_i + \rho(i,z)c(z)],
\]
while the cost function is given by $c(z)$. 

\medskip
\noindent
This specification is motivated as follows. 

\medskip
-- In the absence of public projects, Hervés-Beloso et al. (2008) consider a splitting of the initial endowment along the direction of the allocation $(\xi_1, \ldots, \xi_n)$ by assigning to each agent $i$ the bundle $\alpha_i e_i + (1-\alpha_i)\xi_i$,
resulting in a total initial endowment of
\[
\sum_{i =1}^{n}[\alpha_i e_i + (1-\alpha_i)\xi_i].
\]

--In the presence of a public project, however, the allocation direction is adjusted so as to incorporate both the private consumption $\xi_i$ and the individual contribution $\rho(i,z)c(z)$. Hence, the modified endowment of agent $i$ in the auxiliary economy is given by \footnote{We refer to Remark \ref{rem:GR} for a comparison with Graziano and Romaniello (2012).}
\[
\alpha_i e_i + (1-\alpha_i)[\xi_i + \rho(i,z)c(z)],
\]
resulting in a total initial endowment of 
\[
\sum_{i\in N}[\alpha_i e_i + (1-\alpha_i)(\xi_i + \rho(i,z)c(z))].
\]

We now present the main result of this section, which provides a characterization of cost share equilibria based on the veto power of the grand coalition in infinitely many economies.

\begin{theorem}\label{zequivalence}
Let $\mathscr E$ be a finite economy. Suppose that $\sigma$ is a contribution measure with the corresponding cost distribution function $\rho$. Let $(x_1,\cdots,x_n,y)$ be a Pareto optimal allocation of the finite economy $\mathscr E$ and let $z\in \mathscr Y$. Suppose further that $(\gamma_1^z,\cdots, \gamma_n^z, z)$ is a feasible allocation satisfying 
$u_i(\gamma_i^z,z) \ge u_i(x_i,y)$ for all $i\in N$.\footnote{Refer to Remark \ref{rem:gamma}.}
Then $(x_1,\cdots,x_n,y)$ is a cost share equilibrium if and only if it is not $z$-dominated in the economy  $\mathscr{E}(\gamma_1^z,\cdots,
\gamma_n^z, z,\alpha)$ for each $\alpha = (\alpha_1, \cdots, \alpha_n)$ with $0\le \alpha_i \le 1$.
\end{theorem} 

\begin{proof}
Let $(x_1,\cdots,x_n,y)$ be a cost share equilibrium, with $\rho$ being the associated cost distribution function, and let $\pi$ be the corresponding equilibrium price system. Then it can be easily observed that the allocation $(f,y)$, defined by $(f(t), y) = (x_i, y)$ for all $t \in T_i$ and for all $i\in N$, is a cost share equilibrium of the economy $\mathscr E_c$ with individual cost distribution functions defined as $\widehat{\rho}(t,z) = n\rho(i,z)$, for $(t,z) \in T_i\times \mathscr Y$ and all $i\in N$. Suppose, for the sake of contradiction, that the allocation $(x_1, \ldots, x_n, y)$ is $z$-blocked by the grand coalition in the economy $\mathscr{E}(\gamma_1^z,\cdots,
\gamma_n^z, z,\alpha)$. Therefore, there exists some allocation $(\xi_1,\cdots,\xi_n)$ of private commodities such that 
\begin{equation}\label{eqnd1}
\sum\limits_{i\in N} \xi_i + c(z) = \sum\limits_{i \in N}  \alpha_i e_i + \sum\limits_{i \in N}  (1 -\alpha_i)(\gamma_i^z+ \rho(i,z)c(z))
\end{equation}
and $u_i(\xi_i,z) > u_i(x_i, y)$ for all $i \in N$. By Remark \ref{rem:gamma}, we have  
\[
\pi(z) \cdot \gamma_i^z + \rho(i,z) \pi(z) \cdot c(z) = \pi(z) \cdot e_i
\
\]
for all $i \in N$. Further, from $u_i(\xi_i,z) > u_i(x_i, y)$, it follows that
\[
\pi(z) \cdot \xi_i + \rho(i,z) \pi(z) \cdot c(z) > \pi(z) \cdot e_i
\]
holds for all $i \in N$. Thus,
\[
\pi(z) \cdot \xi_i + \rho(i,z) \pi(z) \cdot c(z) > \pi(z) \cdot \gamma_i^z + \rho(i,z) \pi(z) \cdot c(z)
\]
for all $i \in N$.  Therefore, for all $i \in N$,
\[
\pi(z) \cdot (1 - \alpha_i) \xi_i + \rho(i,z) (1 - \alpha_i) \pi(z) \cdot c(z) > \pi(z) \cdot (1 - \alpha_i) \gamma_i^z + \rho(i,z) (1 - \alpha_i)\pi(z) \cdot c(z),
\]
and
\[
\pi(z) \cdot \alpha_i \xi_i + \rho(i,z) \alpha_i \pi(z) \cdot c(z) > \alpha_i \pi(z) \cdot e_i.
\]
Adding these inequalities yields
\[
\pi(z) \cdot \xi_i + \rho(i,z) \pi(z) \cdot c(z) > \pi(z) [\alpha_i e_i + (1 - \alpha_i) \gamma_i^z] + \rho(i,z) (1 - \alpha_i)\pi(z) \cdot c(z)
\]
for all $i \in N$. It follows that, for all $i \in N$, 
\[
\pi(z) \cdot \alpha_i + \rho(i,z) \pi(z) \cdot c(z) > \pi(z) [\alpha_i e_i + (1 - \alpha_i)(\gamma_i^z + \rho(i,z) \pi(z) \cdot c(z))].
\]
Summing over all agents, we obtain
\[
\sum\limits_{i \in N}\pi(z) \cdot \xi_i +  \pi(z) \cdot c(z) > \pi(z)\cdot \sum\limits_{i\in N} [\alpha_i e_i + (1 - \alpha_i)(\gamma_i^z + \rho(i,z) \pi(z) \cdot c(z))],
\]
which contradicts Equation (\ref{eqnd1}).

\medskip
Conversely, assume that $(x_1,\cdots,x_n,y)$ is a non-dominated allocation in the economy $\mathscr{E}(\gamma_1^z,\cdots,
\gamma_n^z, z,\alpha)$. Let $(f,y)$ be a step function on the set $T = [0,1]$, defined by $(f(t), y) = (x_i, y)$ if $t\in T_i$ and $i\in N$. Suppose, by way of contradiction, that $(x_1,\cdots,x_n, y)$ is not a cost share equilibrium allocation with respect to a cost distribution function $\rho$. 
Then one can immediately conclude that the step allocation $(f,y)$
 is not a cost share equilibrium allocation with respect to the cost distribution function $\widehat \rho$ in the associated continuum economy $\mathscr E_c$. By 
Theorem \ref{thm:equivalence} (as well as Theorem \ref{thm:core-wal2} and Theorem \ref{thm:core-wal3}), one can claim that $(f,y) \notin 
\mathscr C^{\widehat \sigma}(\mathscr E_c)$, where $\hat{\sigma}$ is the contribution measure 
obtained from $\widehat{\rho}$.  In view of Theorem \ref{thm:infintevind}, there exists a coalition $S \subseteq T$ with $\mu(S) > 1 - \frac{1}{n}$ $\widehat \sigma$-blocking 
the allocation $(f,y)$ via an allocation $(g,z)$. Therefore, we have $u_t(g(t),z) > u_t(f(t), y)$ $\mu$-a.e. on $S$ and 
\[
\int_{S} g\, d\mu + \widehat{\sigma}(S,z)\widehat{c}(z) = \int_{S} e\, d\mu. 
\]
Define 
\[
g_i := \frac{1}{\mu(S_i)}\int_{S_i}g d\mu,
\] 
where $S_i = S \cap T_i$ for $i\in N$. Then the above equation can be written as
\[
\sum\limits_{i\in N} \mu(S_i)g_i + \sum\limits_{i\in N} \sigma(\{i\},z)\frac{\mu(S_i)}{\mu(T_i)}\frac{c(z)}{n} = \sum\limits_{i\in N} \mu(S_i)e_i.
\]
Define $\alpha _i := \mu(S_i)$ for all $i \in N$. Then
\[
\sum\limits_{i\in N} \alpha_i g_i + \sum\limits_{i\in N} \alpha_i \sigma(\{i\},z)c(z) = \sum\limits_{i\in N} \alpha_i e_i .
\]
Adding $\sum_{i =1}^{n}(1-\alpha_i)\sigma(\{i\},z)c(z)$ on both sides of the above equation we get
\[
\sum\limits_{i\in N} \alpha_i g_i + c(z) = \sum\limits_{i\in N} [\alpha_i e_i + (1-\alpha_i)\sigma(\{i\},z)c(z)].
\]
 Define $h_i = \alpha_i g_i + (1-\alpha_i)\gamma_i^z$. Then in view of the above equation, it follows that
\[
\sum\limits_{i\in N} h_i + c(z) = \sum\limits_{i \in N}  \alpha_i e_i + \sum\limits_{i \in N}  (1 -\alpha_i)[\gamma_i^z + \sigma(\{i\},z)c(z)].
\]
Since $u_i(h_i,z) > u_i(x_i,y)$ for all $i \in N$, we conclude that $(x_1,\cdots,x_n,y)$ is $z$-dominated in $\mathscr{E}(\gamma_1^z,\cdots,
\gamma_n^z, z,\alpha)$. This is a contradiction. 
\end{proof}

\begin{remark}\label{rem:GR}
Our construction of endowments in the auxiliary economy differs from the approach adopted by Graziano and Romaniello (2012). In their formulation, each agent in the finite economy initially pays for public good provision according to a uniform cost share, and then the initial endowments are adjusted in the direction of $(\xi_1,\cdots, \xi_n)$.\footnote{By contrast, our construction allows for the possibility of non-uniform cost shares, thereby generalizing the framework in Graziano and Romaniello (2012).} The total market endowment in their model is given by
\[
\sum_{i \in N} \alpha_i \left(e_i - \frac{1}{n}c(z)\right) + (1-\alpha_i)\xi_i,
\]
which motivates their cost function
\[
c_{\alpha}(z) = \sum_{i \in N} \alpha_i \frac{c(z)}{n},
\]
and the initial endowment allocation
\[
\alpha_i e_i + (1-\alpha_i)\xi_i
\]
for each \(i \in N\). Using the aggregate smoothness condition (see page 291 in Graziano and Romaniello (2012)), together with our assumptions, they characterize cost share equilibria of \(\mathscr{E}\) in terms of \(z\)-domination as in Theorem \ref{zequivalence}. In contrast, our approach modifies only the initial endowment allocation rather than the cost function itself, resulting in a different initial endowment structure from that proposed in Graziano and Romaniello (2012).
\end{remark}

The following theorem summarizes our major findings for a finite economy with public projects where the space of private commodities is a Banach lattice. 

\begin{theorem}\label{thm4}
Let $\mathscr E$ be a finite economy. Suppose that $\sigma$ is a contribution measure with the corresponding cost distribution function $\rho$. Let $(x_1,\cdots,x_n,y)$ be a Pareto optimal allocation of the finite economy $\mathscr E$ and let $z\in \mathscr Y$. Suppose that $(\gamma_1^z,\cdots, \gamma_n^z, z)$ is a feasible allocation satisfying 
$u_i(\gamma_i^z,z) \ge u_i(x_i,y)$ for all $i\in N$. Further, assume that $e_{i} - \sigma(\{i\},z)c(z) \gg 0$ holds  for each agent $i \in N$, and $z \in \mathscr{Y}$.
 Then the following statements are equivalent:

\begin{itemize}

\item[\emph{(i)}] $(x_1, \cdots,x_n, y)$ is a cost share equilibrium allocation.

\item[\emph{(ii)}] $(x_1, \cdots,x_n, y)$ is in the $\sigma$-Aubin core of the economy.

\item[\emph{(iii)}] $(x_1, \cdots,x_n, y)$ is an Aubin non-dominated allocation.

\item[\emph{(iv)}] $(x_1,\cdots,x_n, y)$ is a $\sigma$-Edgeworth equilibria.

\item[\emph{(v)}] $(x_1, \cdots,x_n, y)$ is not $z$-dominated in the economy  $\mathscr{E}(\gamma_1^z,\cdots,
\gamma_n^z, z,\alpha)$ for each $\alpha = (\alpha_1, \cdots, \alpha_n)$ with $0\le \alpha_i \le 1$.

\end{itemize}

\end{theorem}

\begin{proof}
The fact that (i) $\Longleftrightarrow$ (ii) follows from Theorem \ref{fuzzyequivalence} whereas (ii) $\Longleftrightarrow$ (iv) and (i) 
$\Longleftrightarrow$ (v) follows from Proposition \ref{eeequivalence} and Theorem \ref{zequivalence}, respectively. The fact that (ii) $\implies$ (iii) follows from the definition. 
To establish that the above statements are equivalent, all that remains to show is (iii) $\implies$ (ii). Let $(x_1, \cdots,x_n, y)$ be not in the $\sigma$-Aubin core of the economy. Then the corresponding allocation $(f,y)$ does not belong to the $\widehat{\sigma}$-core of the associated continuum economy $\mathscr E_C$, where 
\[
\widehat{\sigma}(S,z) = \sum\limits_{i =1}^{n} \sigma(\{i\},z)\frac{\mu(S \cap T_i)}{\mu(T_i)}.
\] 
By Theorem \ref{thm:infintevind}, there exists a coalition $S$ that blocks the allocation $(f,l)$ such that $\mu(S) > 1 - \frac{1}{n}$ via an allocation $(g,z)$. Then there exists a 
coalition $\gamma \in \mathcal A_{c}$, a public project $z \in \mathscr Y$, and a private good assignment $(\xi_1,\cdots,\xi_n)$ such that
\begin{equation}\label{eqnf3}
\sum\limits_{i\in N} \gamma_i \xi_i + \widetilde{\sigma}(\gamma,z)c(z) = \sum\limits_{i\in N}\gamma_i e_i
\end{equation}
and 
\[
u_i(\xi_i, z) > u_i(x_i, y), \mbox{ for all } i \in {\rm supp}\gamma.
\]
where $\gamma_i := \mu(S \cap T_i)> 0$ and 
\[
\xi_i=\frac{1}{\gamma_i}\int_{S\cap T_i} g\, d\mu,
\]
for all $i\in N$. Hence, $(x_1, \cdots, x_n, y)$ is not a Aubin non-dominated allocation in the finite economy $\mathscr E$. This establishes the claim.
\end{proof}

\section{Appendix}\label{appendix}

{\bf Proof of Theorem \ref{fuzzyequivalence}:} We first show that any allocation $(x_1,\cdots,x_n,y) \in CE_{\rho}(\mathscr E)$ belongs to 
the Aubin $\sigma$-core of the economy. Suppose not, then there exists a 
coalition $\gamma \in \mathcal A_{c}$, a public project $z \in \mathscr Y$, and an allocation of private goods $(\xi_1,\cdots,\xi_n)$ such that
\begin{equation}\label{eqnf1}
\sum\limits_{i\in N} \gamma_i \xi_i + \widetilde{\sigma}(\gamma,z)c(z) = \sum\limits_{i\in N}\gamma_i e_i
\end{equation}
and 
\[
u_i(\xi_i, z) > u_i(x_i, y), \mbox{ for all } i \in {\rm supp}\gamma.
\]
Further, let $\pi$ denote a price system associated with the cost share equilibrium allocation $(x_1,\cdots,x_n,y)$. Then, from the definition of cost share equilibrium 
it follows that, for all $i \in {\rm supp}\gamma$,
\[
\pi(z) \cdot \xi_i +  \rho(i,z) \pi(z) \cdot c(z) > \pi(z) \cdot e_i.
\]
Consequently,
\[
\pi(z) \cdot \gamma_i \xi_i + \gamma_i  \rho(i,z) \pi(z) \cdot c(z) > \pi(z) \cdot \gamma_i e_i
\]
for each $i\in {\rm supp}\gamma$. Since $\widetilde{\sigma}(\gamma,z) = \sum\limits_{i\in N}\gamma_i \rho(i, z)$, it follows that
\[
\pi(z) \cdot \sum\limits_{i\in N} \gamma_i \xi_i +  \widetilde{\sigma}(\gamma,z) \pi(z) \cdot c(z) > \pi(z) \cdot \sum\limits_{i\in N} \gamma_i e_i .
\]
This is a contradiction to Equation (\ref{eqnf1}) and thus our claim follows.

\medskip
\noindent
We now establish the converse and show that an allocation $(x_1,\cdots,x_n, y)$ in the Aubin $\sigma$-core of the economy $\mathscr E$ is a cost share equilibrium 
of the economy $\mathscr E$. To begin with, consider the associated allocation $(f,y)$ in the continuum economy $\mathscr E_c$ defined by $f(t) = x_i$, for all $t \in T_i$
and all $i\in N$. 

 \medskip
\noindent
{\bf Claim I: $\mathbf{(f,y)}$ belongs to the $\mathbf{\hat{\sigma}}$-core of  $\mathbf{\mathscr E_c}$.} Suppose not, then there would exist a coalition
$S \in \mathscr T$ with $\mu(S) > 0$, and an allocation $(g,z)$ such that
\[
\int_{S} g d\mu + \hat{\sigma}(S,z)\hat{c}(z) = \int_{S} e d\mu,
\]
and
\[
u_t(g(t),z) > u_t(f(t),y), \mbox{ $\mu$-a.e. on } S.
\]
In view of the above inequality, one obtains that
$u_i(g(t), z) > u_i(x_i, y)$ $\mu$-a.e. on $S \cap T_i$ and all $i\in N$.
Let $\mathbb J = \{ i \in N: \mu(S \cap T_i) > 0\}$. Define
\[
\gamma_i = \left\{
\begin{array}{ll}
\mu(S \cap T_i), \mbox{ if } i \in \mathbb J, \\
0, \hspace{42pt} \mbox{ otherwise.}
\end{array}
\right.
\]
 For $i \in \mathbb J$, define 
\[
g_i: = \frac{1}{\gamma_i}\int_{S \cap T_i}g\, d\mu
\] 
for all $i\in \widetilde{N}$. Thus, it follows that
\[
\sum\limits_{i\in \mathbb J} \gamma_i g_i = \int_{S} g\, d\mu = \int_{S} e\, d\mu - \hat{\sigma}(S,z)\hat{c}(z) = \sum\limits_{i \in \mathbb J}\int_{S \cap T_i} e\, d\mu - 
\sum\limits_{i \in \mathbb J}\sigma(\{i\}, z)\frac{\mu(S \cap T_i)}{\mu(T_i)}\frac{c(z)}{n}
\]
It can be noted that the above equation simplifies to
\[
\sum\limits_{i \in N} \gamma_i g_i = \sum\limits_{i\in N} \gamma_i e_i - \widetilde{\sigma}(\gamma,z)c(z).
\]
By virtue of Jensen's integral inequality, one can argue that $u_i(g_i, z) > u_i(x_i, y)$ for all $i \in \mathbb J$, which yields a contradiction.

\medskip
By {\bf Claim I} and Theorem 
\ref{thm:equivalence}, $(f,y)$ is a cost share equilibrium with respect to the individual cost function $\widehat{\rho}$ associated to $\rho$ and defined by 
$\hat{\rho}(t, z) = n\rho(i,z)$, for all $(t,z) \in T_i\times \mathscr Y$ and each $i\in N$. Let $\pi: \mathscr Y \to \mathbb B^*$ be the price system associated to 
the cost share equilibrium allocation $(f,y)$. 

\medskip
\noindent
{\bf Claim 2: $\mathbf{(x_1,\cdots,x_n, y)}$ is a cost share equilibrium for the price system $\mathbf{\pi}$.} For each $i \in N$ and $t \in T_i$, we obatin
\[
\pi(y) \cdot x_i +  \rho(i,y) \pi(y)\cdot c(y) = \pi(y) \cdot f(t) +  \hat{\rho}(t,y) \pi(z) \cdot \hat{c}(y) \le \pi(y) \cdot e(t).
\]
Pick an $i \in N$, $z \in \mathscr Y$, and $g \in \mathbb B_+$ such that $u_i(\xi, z) > u_i(x_i, y)$ holds. As $(f,y)$ is a cost share equilibrium, the following holds for almost all $t \in T_i$, 
\[
\pi(z) \cdot \xi + \hat{\rho}(t,z) \pi(z)\cdot \hat{c}(z) > \pi(z) \cdot e(t).
\]
Hence, it follows that
\[
\pi(z) \cdot \xi + \rho(i,z) \pi(z) \cdot c(z) > \pi(z) \cdot e_i.
\]
This establishes our claim, and the conclusion follows.

\medskip
\noindent
{\bf Proof of Theorem \ref{eeequivalence}:}   Let the allocation $(x_{1},\cdots,x_{n},y)$ belong to the $\sigma$-Aubin core  of $\mathscr E$ but is not a $\sigma$-Edgeworth equilibrium. So, there exists an $r$-replica 
$\mathscr{E}_{r}$, a coalition $C$ of $\mathscr{E}_{r}$, and an allocation $ (\xi_{(1,1)},\cdots,\xi_{(1,r)},\cdots,\xi_{(n,1)},\cdots,\xi_{(n,r)}, z)$ such that 

\[
u_{(i,j)}(\xi_{(i,j)},z)=u_{i}(\xi_{(i,j)},z)>u_{i}(x_{i},y)
\]
for all $(i,j) \in C$ and
  \[
        \sum\limits_{(i,j) \in C} \xi_{(i,j)} + \sum\limits_{(i,j) \in C} \sigma(\{(i,j)\},z)c_{r}(z) = \sum\limits_{(i,j) \in C} e_{(i,j)}. 
   \]
   Let $C_i:=\{j:(i,j)\in C\}$ and let $l_{i}$ denote the cardinality $C_i$. Define $A:=\{i \in N: l_{i} \neq 0\}$. For each $i \in A,$ define 
\[
\xi_{i}=\sum_{j\in C_i} \frac{1}{l_{i}}\xi_{(i,j)}.
\] 
In view of the above inequality, we obtain
  \[
 \sum\limits_{i \in A} l_{i}\xi_{i} + \sum\limits_{i \in A} l_{i}\frac{\sigma(\{i\},z)}{r}rc(z) = \sum\limits_{i \in A} l_{i}e_{i}.
 \]
 Considering the Aubin coalition $\gamma$ defined by 
   \[
   \gamma_i:= \left\{
  \begin{array}{ll}
   \frac{l_{i}}{\sum_{i\in A} l_{i}},& \mbox{if $i\in A$;}\\[0.5em]
   0,& \mbox{otherwise,}
  \end{array}
  \right.
  \]
 one can arrive at a contradiction following from the convexity of the utility functions.

    Conversely, let $(x_{1},\cdots,x_{n},y)$ be a $\sigma$-Edgeworth equilibrium, and assume, on the contrary, that there exists an Aubin coalition 
    $\gamma$ and an allocation $(\xi_{1},\cdots,\xi_{n},z)$ such that $u_{i}(\xi_{i},z)>u_{i}(x_{i},y)$ for all $i \in {\rm supp}\gamma$ and
    \[
        \sum\limits_{i \in {\rm supp}\gamma} \gamma_{i}\xi_{i} + \widetilde{\sigma}(\gamma,z)c(z) = \sum\limits_{i \in {\rm supp}\gamma} \gamma_{i}e_{i}. 
   \]
Choose some $\varepsilon \in (0,1)$ such that $u_{i}(\varepsilon \xi_{i},z) > u_{i}(x_{i},y),$ for each $i\in N$. Thus, the last inequality can be rewritten in the form
\[
 \sum\limits_{i \in {\rm supp}\gamma} \frac{\gamma_{i}}{\varepsilon} [ \varepsilon \xi_{i} + (1-\varepsilon)(e_{i} - \sigma(\{i\},z)c(z)) ] + \sum\limits_{i \in {\rm supp}\gamma} \frac{\gamma_{i}}{\varepsilon}  \sigma(\{i\},z)c(z) = \sum\limits_{i \in {\rm supp}\gamma} \frac{\gamma_{i}}{\varepsilon}e_{i}.
    \]
    By the monotonicity assumption, one can infer that 
    \[
    u_{i}(\varepsilon \xi_{i} + (1-\varepsilon)(e_{i} - \sigma(\{i\},z)c(z)),z) > u_{i}(x_{i},y),
    \] 
    where $\varepsilon \xi_{i} + (1-\varepsilon)(e_{i} - \sigma(\{i\},z)c(z))$ is a strictly positive vector of private goods. Hence, without loss of generality, one can assume that $\xi_{i} \gg 0$ for each $i\ \in N$ and therefore, by continuity, one obtains
 \[
 \sum\limits_{i \in {\rm supp}\gamma} \gamma_{i}\xi_{i} + \sum\limits_{i \in {\rm supp}\gamma}  \gamma_{i}\sigma(\{i\},z)c(z) \ll \sum\limits_{i \in {\rm supp}\gamma} \gamma_{i}e_{i}.  
    \]
The last inequality ensures that it is possible to replace the Aubin coalition $\gamma$ by a rational-valued coalition $\gamma^{\prime}$ in such a way that the above inequality still holds.

\medskip
    Let $r$ be a positive integer such that $l_{i}=r\gamma_{i}^{\prime}$ is  a positive integer, for every $i \in  {\rm supp}\gamma$. Define the coalition $S$ in $r$-fold replica $\mathscr{E}_{r}$ of $\mathcal{E}$ as the coalition containing agents $(i,j)$, $j=1,...,l_{i}$, and for $i \in$ supp$\gamma.$ Define $\xi_{(i,j)}:= \xi_{i},$ for $j=1,\cdots,l_{i}$, for each $i \in {\rm supp}\gamma$. It follows from the previous inequality that
\[
 \sum\limits_{i \in {\rm supp}\gamma} l_{i} \xi_{(i, j)}+\sum\limits_{i \in  {\rm supp}\gamma} l_{i} \sigma(\{(i, j)\},z) c_{r}(z) \ll \sum\limits_{i \in 
 {\rm supp}\gamma} l_{i} e_{(i, j)}.
 \]
 This is a contradiction to the fact that the allocation $(x_{1},\cdots, x_{n},y)$ belongs to the $\sigma$-core of the economy $\mathscr{E}_{r}.$

\bigskip
\noindent
{\bf Acknowledgement:} The author would like to thank the Core Research Grant (Grant No. CRG/2023/009012) issued to him under the SERB scheme of Anusandhan National Research Foundation. The funding support has helped him enormously with the research outcome.

\end{document}